\documentclass{article}
\usepackage[margin=1in]{geometry}

\usepackage{fancyhdr}

\usepackage[utf8]{inputenc}
\usepackage{lmodern}
\usepackage[T1]{fontenc}

\usepackage{amsmath, amssymb, amsthm, amsfonts}

\usepackage[polish, german, english]{babel}

\usepackage{bbm}
\usepackage[retainorgcmds]{IEEEtrantools}
\usepackage{thm-restate}
\usepackage{enumerate}

\usepackage{tikz}
\usetikzlibrary{calc, arrows, decorations.markings, positioning, fit,
shapes.geometric, arrows.meta}
\tikzset{>=stealth'}
\tikzstyle{vecArrow} = [thick, decoration={markings,mark=at position
  1 with {\arrow[semithick]{open triangle 60}}},
  double distance=1.4pt, shorten >= 5.5pt,
  preaction = {decorate},
  postaction = {draw,line width=1.4pt, white,shorten >= 4.5pt}]

\usepackage{listings}
\usepackage{verbatim}

\theoremstyle{plain}
        \newtheorem{theorem}{Theorem}
        \newtheorem{lemma}[theorem]{Lemma}
        \newtheorem{corollary}[theorem]{Corollary}
        \newtheorem{claim}[theorem]{Claim}
        \newtheorem{fact}[theorem]{Fact}

\theoremstyle{definition}
        \newtheorem{definition}[theorem]{Definition}

\theoremstyle{remark}
        \newtheorem{remark}[theorem]{Remark}

\DeclareMathOperator{\Span}{span}
\DeclareMathOperator{\RM}{RM}
\DeclareMathOperator*{\argmax}{argmax}

\newcommand{\eps}{\varepsilon}
\newcommand{\ta}{\tilde{a}}

\newcommand{\tA}{\tilde{A}}

\newcommand{\cA}{\mathcal{A}}

\usepackage[pdftex]{hyperref}

\author{Emmanuel Abbe\footnote{EPFL and Princeton University,
email: \url{emmanuel.abbe@epfl.ch}.} 
\and Jan Hązła\footnote{EPFL, email: \url{jan.hazla@epfl.ch}.} 
\and Ido Nachum\footnote{EPFL, email: \url{ido.nachum@epfl.ch}.}
}
\date{}

\title{Almost-Reed--Muller Codes Achieve Constant Rates for Random Errors}

\begin{document}

\pagenumbering{gobble}
\maketitle
\begin{abstract}
    This paper considers 
    ``$\delta$-almost Reed--Muller codes'',
    i.e., linear codes spanned by evaluations of all
    but a $\delta$ fraction of monomials of degree at
    most $d$. It is shown that for any $\delta > 0$ and any $\eps>0$, there exists a
    family of $\delta$-almost Reed--Muller codes of
    constant rate that correct $1/2-\eps$ fraction of 
    random errors with high probability. For exact Reed--Muller codes, the analogous result is not known and represents a weaker version of the longstanding conjecture that Reed--Muller codes achieve capacity
    for random errors (Abbe-Shpilka-Wigderson STOC '15). 
    Our proof is based on the recent polarization result for 
    Reed--Muller codes, combined with a combinatorial approach to establishing
    inequalities between the Reed--Muller code entropies.
\end{abstract}
\clearpage
\pagenumbering{arabic}

\section{Introduction}
Reed--Muller (RM) codes \cite{Reed54,Muller54} have long been conjectured to achieve the Shannon capacity for symmetric channels.
Traces of this conjecture date back to the 1960s, as discussed recently in~\cite{survey_RM} (we refer to \cite{Kudekar17,survey_RM} for further references).  

In the recent years, significant progress has been made on this conjecture for both the binary erasure channel (BEC), a.k.a.\ random erasures, and for the binary symmetric channel (BSC), a.k.a.\ random errors, using a variety
of approaches. 
These are based on (i) estimating the weight enumerator, verifying the conjecture for some vanishing rates on the BEC and BSC~\cite{Abbe15,Sberlo18} --- significant improvements on the weight enumerator were also recently obtained for certain rates
using contractivity arguments~\cite{Samorod18}; (ii)
studying common zeros of bounded degree polynomials, verifying the conjecture for some rates tending to one on the 
BEC~\cite{Abbe15}; (iii) using sharp threshold results for monotone Boolean functions~\cite{Kahn88,Bourgain92}, settling 
the conjecture for the BEC at constant rate~\cite{Kudekar17}, a major step towards the general conjecture.  

The main conjecture of achieving capacity on the BSC in the constant rate regime remains nonetheless open, and currently this does not seem reachable from the above developments. 
In fact, the weaker conjecture of achieving both a constant code rate and a constant error rate (i.e., showing a ``good code'' property in the average-case sense) had been until
recently still open~\cite{Abbe15,survey_RM}. More specifically: 
\begin{center}
Can RM codes of rate $\Omega(1)$ decode an $\Omega(1)$
fraction of random errors with high probability?  
\end{center}

In this paper, we prove the above for almost-RM codes. To make this statement precise, let us make some definitions.   
Let $m\in\mathbb{N}$ and $n:=2^m$. Consider the $n\times n$ \emph{Reed--Muller matrix} $M$ where rows are indexed 
 by subsets $A\subseteq[m]$ and columns are indexed by Boolean vectors $z\in\{0,1\}^m$, and the entries are
given by\footnote{All algebraic operations
are over $\mathbb{F}_2$ unless otherwise specified.}
\begin{align}\label{eq:16}
  M_{A,z}:=\prod_{i\in A}z_i\;,
\end{align}
i.e., they are evaluations of monomials with variables indexed by $A$. 
We also use the notation $v_A$ for the row  $M_{(A,\cdot)}$. Let $\mathcal{A}$ be a collection of subsets of $[m]$. Define the linear code
\begin{align*}
  \RM(m, \cA):=\Span\{v_A: A\in\cA\}\;.
\end{align*}
The special case of $\mathcal{A}$ corresponding to all subsets of size at most $r$
gives the Reed--Muller code 
$\RM(m, r)$ of degree $r$ on $m$ variables.

\begin{definition}\label{def:delta-almost}
  We say that a code $\RM(m,\mathcal{A})$
  is a \emph{$\delta$-almost Reed--Muller code} if there
  exists $r$ such that 
  $\RM(m,\mathcal{A})\subseteq\RM(m,r)$
  and $|\mathcal{A}|\ge(1-\delta)\sum_{i=0}^r\binom{m}{i}$.
\end{definition}
In other words, a $\delta$-almost RM code can
be constructed by deleting up to $\delta$ fraction of
vectors from the standard basis of an RM code.
Note that for an RM code
of rate $R$, a $\delta$-almost RM code
has rate at least $(1-\delta)R$. 
Our result can be stated as:
\begin{theorem}[Cf.~Theorem~\ref{thm:main}]
\label{thm:main-simple}
For any $\delta > 0$ and any $\eps > 0$,
there exists $R > 0$ and a
family of $\delta$-almost Reed--Muller codes of rate $R$ that correct $1/2-\eps$ fraction of random errors with high probability. 
\end{theorem}

To the best of our knowledge, such a result cannot be deduced from the known relationship between RM codes and polar codes \cite{Has13}, nor from the previous work on polarization of
RM codes~\cite{AY19} (see Section~\ref{sec:related}
for more details).
For exact RM codes, the closest result comes from \cite{Sberlo18} (improving upon~\cite{Abbe15}), which shows that for $r < m/2 - \Omega({\sqrt{m\log m}})$, corresponding to rates 
$R=1/\mathrm{polylog}(n)$,
the RM$(m,r)$ codes decode a fraction  $1/2-o(1)$ of random errors with high probability. 

We emphasize that when we say that the codes ``correct''
a fraction of errors we mean that there exists a decoding
algorithm that succeeds with high probability and we do not claim
that this algorithm is efficient. Efficiency of decoding of Reed--Muller
codes is a research area of its own (see, e.g., Section~V
of~\cite{survey_RM}).

Our approach is based on the polarization theory for RM codes \cite{AY19}. Polarization theory emerged from the development of polar codes \cite{Arikan09}, and consists in tracking the conditional entropies of the information bits rather than the block error probability directly. We refer to Sections~\ref{sec:notation}
and~\ref{sec:import} for a precise description of
concepts and results that are used in this work,
including successive decoding and the
relevant decoding order (and \cite{Arikan09, ArikanTelatar,Arikan2010survey,Blasiok18} for references on polarization theory). Now we provide a brief description of the approach.
 
For RM codes on the BSC, the channel output is $Y=U M+E$ where $U$ contains $n$ i.i.d.\ uniform bits (the information bits) and $E$ contains $n$ i.i.d.\ Bernoulli$(p)$ bits (the noise).
Recall that the rows of $M$, and therefore also the coordinates
of $U$, are indexed by subsets $A\subseteq[m]$.
In Definition~\ref{def:decoding-order} we define a certain
total ordering on the subsets, basically from the sparsest to
the heaviest rows of $M$. Then, the conditional entropies
are defined using this ordering by $H_A:=H(U_A|Y,(U_B)_{B<A})$.
To decode with a small error probability, these conditional entropies should be low for the components selected by the code, i.e., those indexed by sets $A\in\mathcal{A}$ in the 
$\RM(m,\mathcal{A})$ code. In particular, to prove that
successive decoding of $\RM(m, r)$ succeeds, we would like to
show that all components indexed by sets with size up to $r$
have low entropies.

In the case of polar codes, a different decoding order is used
(in other words, rows of the matrix $M$ are permuted), which allows for a simpler recursive analysis. In \cite{Arikan09}, the {\it polarization} property is shown for the polar ordering, i.e., it is shown that most conditional entropies are tending to either 0 or 1, and thus that the code induced by the low entropy components achieves capacity.\footnote{One also needs to show that the decay to 0 is fast enough. With polar codes, capacity is obtained for any binary input symmetric output channel, including the BEC and BSC.} In \cite{AY19}, the same property is shown for the RM code ordering, and thus, the code induced by the low entropy components in this ordering also achieves capacity. This code is called the \emph{twin-RM} code in \cite{AY19}.

To prove the capacity conjecture for RM codes, it remains to show that the twin-RM code is in fact the RM code for any symmetric channel. 
In other words, we would like to prove
that the components with low entropies correspond in fact to the small subsets $A$. In particular, showing that 
\begin{align}\label{eq:40}
    |A| > |B| \implies H_A \ge H_B
\end{align}
would be sufficient to prove that the two codes are the same. 
The `twin' terminology is used in \cite{AY19} mainly based on numerical simulations, which provide some evidence that the codes are indeed the same, and on a proof of their equality up to blocklength $n=16$. Moreover,~\cite{AY19} shows a weaker
property
\begin{align}\label{eq:38}
    A\supseteq B\implies H_A\ge H_B\;,
\end{align}
indicating that at least in general the entropies increase
with the set size.
However, \cite{AY19} does not prove that the twin-RM code is in fact the RM code and~\eqref{eq:38} does not seem to be enough
for a good bound on their similarity.

In this paper, we push this approach and study how close the twin-RM code is to the RM code, investigating what rates can be achieved using the level of similarity that we can quantify. Rewriting~\eqref{eq:38} 
as $H_{A}\ge H_{A\setminus\{a\}}$ for any $a$,
one of our contributions is to establish a stronger property
\begin{align}\label{eq:39}
    B=(A\cup\{b\})\setminus\{a,a'\}
    \implies H_A\ge H_B
\end{align}
for any $a,a'\in A$ and any $b$. Subsequently, 
\eqref{eq:39} is combined with other inequalities
that follow from symmetries of RM codes and the decoding order.
By arguments featuring deviation bounds on integer random walks,
we then show that while we are still short of~\eqref{eq:40},
we have that $H_A\ge H_B$ holds for \emph{almost all} pairs of sets
$(A, B)$ with $|A|=r+k, |B|=r$ for some relevant values
of $r$ and $k$.
This allows us to use a standard polarization argument based on 
the relation $\sum_A H_A=h(p)\cdot n$, where $1-h(p)$ is the 
channel capacity, to conclude
that for the BSC and certain rates twin-RM codes and RM codes are indeed ``close cousins'',
resulting in Theorem~\ref{thm:main-simple}.

Consider the channel naturally given by a conditional entropy $H_A$. 
That is, consider the channel $Y,U_{<A}\to U_A$
which takes a noisy codeword $Y$
and information bits $U_C$ for sets $C$ preceding $A$ in the
decoding order and outputs the information bit $U_A$.
The proof of $H_A\ge H_B$ in~\eqref{eq:39} in fact establishes that such 
channel for set $A$ is a \emph{degradation} \cite{Cover72, Bergmans73, MP18}
of the channel for $B$.
More specifically, we use the relation between channel and source
coding for BSC to show that the channel for $A$ can be obtained
by applying a linear isomorphism on the inputs of 
the channel for $B$ and subsequently dropping some information
bits. This isomorphism is induced by
a permutation of columns of the RM matrix $M$.


We conclude this discussion by noting that some of the mechanisms we develop in this paper (e.g., Lemma~\ref{lem:expansion})
could be reused to further improve 
the rate/$\delta$-closeness tradeoff if more inequalities on the RM entropies were to be obtained.
In Section~\ref{sec:lower-bound} we show that certain
aspects of our analysis are tight, suggesting that
more inequalities are also necessary in order to advance
our approach. A natural next step would be to 
extend~\eqref{eq:39} to $H_A\ge H_B$ in case where $B$
can be constructed from $A$ by removing \emph{three}
and adding \emph{two} elements. However, at the moment we
do not know how to do that.
We stress that the full inequality~\eqref{eq:40} might well be true.
Circumstantial evidence in its favor is provided by simulations conducted in~\cite{AY19}
in the context of erasure channel, see Figures~11--13 therein.

Another natural question to investigate is if there exists a way to ``boost'' a
good $\delta$-almost RM code and obtain a statement about actual RM codes.
For example, using symmetry arguments
we can show that if an RM code with two vectors deleted from the standard
basis is good for a BSC channel, then the full RM code is also good.
However, this is a far cry from a constant $\delta$ fraction that appears
in our results.

\subsection{Quantifying closeness to Reed--Muller codes}
\label{sec:related}
As discussed, both polar codes and RM codes are 
instances of $\RM(m,\mathcal{A})$ codes,
in other words they are both generated
by some of the rows of the RM matrix $M$.
In this section we briefly discuss
some results on $\delta$-almost RM codes
that are implicit in existing literature on
polarization theory~\cite{Has13, AY19}.
The codes obtained from these works 
can be described as ``distant cousins''
(constant rate for some $\delta>0$),
as opposed to our ``close cousins''
(constant rate for any $\delta>0$).

There is a wealth of literature on performance
of RM codes on different channels
\cite{Helleseth04, Helleseth05, Abbe15, Kudekar17, Saptharishi17, Sberlo18, Samorod18, AY19, survey_RM}
and our foregoing discussion already touched
the ones that are most relevant to this work.
To the best of our knowledge, none of those results
concerning exact RM codes can be easily modified to
obtain good almost-RM codes.

One way to construct a constant rate $\delta$-almost RM code
for the BSC of capacity $C$ 
is to take the intersection of
the pure RM code of rate $R<C$ and a polar code of rate $C-o(1)$.
Since polar codes achieve capacity and since this code is a subset of
the polar code, its successive decoding (in the polar code order)
corrects errors with high probability. On the other hand,
Corollary~4.2 in~\cite{Has13} states that such code will be
asymptotically $\delta$-almost RM for 
$\delta=\delta(R, C)=1-C$. 
Let us emphasize the differences between our work and this result.
Basically, the order of quantifiers is reversed.
\cite{Has13} implies that for every $p$ 
there exists \emph{some} 
constant rate and $\delta$ such that 
$\delta$-almost RM codes
correct fraction $p$ of errors. 
Moreover, this $\delta=1-C$ goes to 1 as
the channel capacity goes to 0.
In contrast,
we show existence of constant rate $\delta$-almost RM codes correcting
fraction $p$ of errors for \emph{every} $\delta>0$
(albeit with rates $R(\delta)\to 0$ as $\delta\to 0$).
On the other hand,~\cite{Has13} gives rates
up to $R=(1-\delta) C=C^2$,
significantly better than what we achieve 
for small $\delta$
(we discuss our rates in Section~\ref{sec:theorem}).
For example, for capacity $C=1/2$~\cite{Has13} 
gives a $1/2$-close RM code of rate $1/4$, while
our construction gives a $\delta$-close code
for any $\delta>0$ and a small, positive rate
$R(\delta)>0$.

While we omit the details here, it turns out
that a result that is virtually identical to the one
we just described can also be obtained
using only 
the relation~\eqref{eq:38} which was already proved in~\cite{AY19}. In contrast, our new result requires a more involved analysis using~\eqref{eq:39}.

\paragraph{Subsequent work} A
subsequent work by one of the authors,
Samorodnitsky and Sberlo~\cite{SS20} shows that
Reed--Muller codes of constant rate $R$
decode errors on BSC$(p)$ for 
$p<\frac{1}{2}-\sqrt{2^{-R}(1-2^{-R})}$, achieving the positive rate conjecture with techniques 
building on papers by
Samorodnitsky~\cite{Samorod18,Sam20}.  This approach is disjoint
from our approach, has the advantage of
dealing with the exact Reed--Muller code, but is no longer related to the successive decoder and to polarization. It remains open to show that constant-rate Reed--Muller codes achieve
capacity on the binary symmetric channel.

\section{Background and Notation}\label{sec:notation}

\subsection{Our setting}

Let $0<p<1/2$ be a parameter of the BSC. The basic elements of our probability
space are two random (row) vectors
\begin{align*}
  U = (U_A)_{A\subseteq [m]}\;,\qquad
  E = (E_z)_{z\in\{0,1\}^m}\;.
\end{align*}
Vector $U$ consists of $n=2^m$ i.i.d.~uniform $\{0,1\}$ random variables
$U_A$ indexed by subsets of $[m]$. The components of vector $E$ are
$n$ i.i.d.~Ber$(p)$ random variables $E_z$ indexed by bitstrings
$z=(z_1,\ldots,z_m)\in\{0,1\}^m$. Furthermore, $U$ and $E$ are independent.

We will consider a total order on sets that we call the \emph{decoding order}:
\begin{definition}[Decoding order]\label{def:decoding-order}
  We say that $A < B$ if
  \begin{itemize}
  \item $|A| > |B|$, or
  \item $|A| = |B|$ and there exists $i$ such that $i\notin A$, $i\in B$
    and $\forall j>i:j\in A\iff j\in B$.
  \end{itemize}
\end{definition}
Note that when $|A|=|B|$, the decoding order can be described as
\emph{reverse lexicographic}: First, all sets that do not contain $m$ come
before all sets that contain $m$. Then, inside of each group,
sets that do not contain $m-1$ come before sets that contain $m-1$,
and so on, recursively. 
To dispel doubts, let us write the decoding order in the case $m=4$ (note the highlighted segment pointing out a difference
from a more ordinary lexicographic ordering):  
\begin{align*}
    1234<123<124<134<234<12<13<
    \mathbf{23<14}
    <24<34<1<2<3<4<\emptyset\;.
\end{align*}
We define another random vector $X = (X_z)_{z\in\{0,1\}^m}$ as
$X:=UM$. More precisely,
\begin{align}\label{eq:17}
  X_z = \sum_{A\subseteq[m]} U_A\prod_{i\in A}z_i\;.
\end{align}
We also let $Y:=X+E$.

We are going to use vectors $U,E,X$ and $Y$ in the
analysis of successive decoding of codes $\RM(m,\mathcal{A})$.
To this end, we will consider several information measures like conditional
entropy $H(X\mid Y)$, Bhattacharyya parameter and MAP decoding error:
\begin{definition}[Bhattacharyya parameter]
  Let $(V,W)$ be discrete random variables such that $V$ is uniform in
  $\{0,1\}$ and $W\in\mathcal{W}$. The Bhattacharyya parameter $Z(V\mid W)$ is
  \begin{align}\label{eq:14}
    Z(V\mid W)
    &:=\sum_{w\in\mathcal{W}}
      \sqrt{\Pr[W=w\mid V=0]\Pr[W=w\mid V=1]}\;.
  \end{align}
\end{definition}
\begin{definition}[MAP error]
  Let $(V,W)$ be discrete random variables such that $V$ is uniform in
  $\{0,1\}$ and $W\in\mathcal{W}$. The maximum a posteriori probability
  (MAP) decoding error $P_e(V\mid W)$ is
  \begin{align*}
    P_e(V\mid W)
    &:=\frac{1}{2}\sum_{w\in\mathcal{W}}
      \min\big(\Pr[W=w\mid V=0],\Pr[W=w\mid V=1]\big)\;.  
  \end{align*}
\end{definition}
Note that $P_e(V\mid W)$ is the probability of error under
an optimal scheme for guessing the value of $V$ given $W$.

Turning back to our setting, let $U_{<A}:=(U_B)_{B<A}$.
When analyzing codes $\RM(m,\mathcal{A})$, we will be interested
in values like
\begin{align*}
  H_A:=H(U_A\mid U_{<A},Y)\;,\qquad
  Z_A:=Z(U_A\mid U_{<A},Y)\;,
\end{align*}
i.e., in information measures encountered in the process of
\emph{successive decoding}: 
Decoding the information bit $U_A$ given the noisy codeword
$Y$ and previously decoded (in the decoding order) bits $U_{<A}$.

For a code $\RM(m,\mathcal{A})$, 
the successive decoding algorithm
decodes a noisy codeword $y$ to $\hat{u}=(\hat{u}_A)_{A\in\mathcal{A}}$,
where each bit $\hat{u}_A$ is guessed according to the MAP formula,
assuming that the preceding bits were decoded correctly and $U_{A}=0$ for
$A\notin \mathcal{A}$:
\begin{align*}
  \hat{u}_A(y):=\argmax_{u\in\{0,1\}}
  \Pr\Big[Y=y,(U_B=\hat{u}_B)_{B<A,B\in\mathcal{A}}, (U_B=0)_{B<A,B\notin\mathcal{A}}
  \mid U_A=u\Big]\;.
\end{align*}

\begin{definition}[Successive decoding under decoding order]
\label{def:successive decoding}
  Using the notation above, the decoding error of
  $\RM(m,\mathcal{A})$ under successive decoding is given
  by
  \begin{align}\label{eq:37}
      \Pr\left[\hat{u}(Y)\ne(U_A)_{A\in\mathcal{A}}\right]\;.
  \end{align}
\end{definition}
In particular, if for a given channel and a code family $\RM(m,\mathcal{A})$
the probability in~\eqref{eq:37} vanishes, then the code
$\RM(m, \mathcal{A})$ corrects errors under this channel with high probability.

\subsection{Miscellaneous notation}

We specify some shorthand notation that we use throughout.
$[m]$ denotes the set $\{1,\ldots,m\}$ and $\mathcal{P}(m)$ is a set
of all subsets of $[m]$. Given $A\subseteq[m]$ we write
$\overline{A}:=[m]\setminus A$. We use binomial coefficients $\binom{m}{k}$
and write $\binom{S}{k}$ for the set of all subsets of $S$ of size $k$.
To avoid clutter we abuse notation writing
$\binom{m}{k}:=\binom{[m]}{k}$. From the context it should always be clear
whether $\binom{m}{k}$ is meant as a number or a set.
We also write
\begin{align*}
  \binom{m}{\le k}:=\sum_{i=0}^k\binom{m}{i}\;,\qquad
  \binom{m}{\le k}:=\binom{[m]}{\le k}:=\bigcup_{i=0}^k\binom{m}{i}
\end{align*}
and analogously for $\binom{m}{\ge k}$. We use $\log$ to denote
binary logarithm and $h(p)$ for the binary entropy function
\begin{align*}
  h(p)=-p\log p-(1-p)\log(1-p)\;.
\end{align*}
We also write $\Phi$ and $\Phi^{-1}$ for the standard Gaussian CDF and its
inverse.
Whenever we consider two sets $A,B$, we try to stick to the
convention that $A$ precedes $B$ in the decoding
order, in particular $|A|\ge|B|$.
We sometimes drop parentheses for consecutive set operations.
In that case we adopt left-to-right associativity, e.g.,
$A\setminus\{a\}\cup\{b\}
=\big(A\setminus\{a\}\big)\cup\{b\}$.

\section{Preliminaries}\label{sec:import}

In this section we list some known results that we use in our proofs.
The most important one is the polarization theorem for RM codes
from~\cite{AY19}:

\begin{theorem}[Theorem~1 in~\cite{AY19}]
  \label{thm:rm-polarization}
  For every $0<p<1/2$, $0<\eps<1/10,c\in\mathbb{N}$ and $0<\xi<1/2$,
  there exists
  $m_0=m_0(p,\eps,c,\xi)$ such that for $m>m_0$,
  \begin{align*}
    \left|\left\{A\subseteq[m]:
    Z_{A}\ge\frac{1}{n^c}\land H_{A}\le 1-\eps
    \right\}\right|
    \le \frac{n}{m^{1/2-\xi}}\;.
  \end{align*}
\end{theorem}
Theorem~\ref{thm:rm-polarization} has
the following interpretation:
The fraction of subsets $A$ for which their
respective channels are not polarized
(where ``polarized'' means that
either $Z_A\approx 0$ or 
$H_A\approx 1$) goes to $0$
at a rate $O(m^{-1/2+\xi})$.
We also use an ingredient from the proof of Theorem~\ref{thm:rm-polarization}:
\begin{lemma}[Lemma~3 in~\cite{AY19}]\label{lem:z-containment}
  If $A\supseteq B$, then $Z_A\ge Z_B$.
\end{lemma}

It is well-known that for symmetric channels (in particular for the BSC)
the probability of error under successive decoding is controlled by 
the sum of the Bhattacharyya parameters:
\begin{theorem}[see~Proposition~2 and Theorem~4 in~\cite{Arikan09}]
  \label{thm:decoding-bound}
  The probability of
  error under successive decoding of $\RM(m,\mathcal{A})$
  is bounded by
  \begin{align*}
    \Pr\big[\hat{u}(Y)\ne (U_A)_{A\in\mathcal{A}}\big]\le \sum_{A\in\mathcal{A}} Z_A\;.
  \end{align*}
\end{theorem}

We also state some known (e.g.,~\cite{Arikan09}) background facts that will be needed in our proofs:
\begin{fact}\label{fac:rm-matrix-inverse}
  The inverse (over $\mathbb{F}_2$) of the Reed--Muller matrix $M$ is given by
  \begin{align*}
    (M^{-1})_{z,A}=\prod_{i\notin A}(1-z_i)\;.
  \end{align*}
\end{fact}
\begin{proof}
  We check directly that
  \begin{align}\label{eq:26}
    (MM^{-1})_{A,B}=\sum_{z\in\{0,1\}^m} \prod_{i\in A}z_i\prod_{i\notin B}(1-z_i)\;.
  \end{align}
  First, if $A=B$, then the sum~\eqref{eq:26}
  has exactly one non-zero term with $z$ being
  indicator of $A$.
  On the other hand, if $A\setminus B\ne\emptyset$, then all terms of the sum
  in~\eqref{eq:26}
  are zero. Furthermore, if $B\setminus A\ne\emptyset$, the number of
  non-zero terms in~\eqref{eq:26} must be even. Hence,
  \begin{align*}
  &
    (MM^{-1})_{A,B}=1\iff A\setminus B=\emptyset\land B\setminus A=\emptyset
    \iff A=B\;.
    \qedhere
  \end{align*}
\end{proof}

\begin{fact}\label{fac:entropy-sum}
  $\sum_{A\subseteq[m]} H_A = h(p)\cdot n$.
\end{fact}
\begin{proof}
  Recall that $X=UM$.
  By Fact~\ref{fac:rm-matrix-inverse}, 
  random vectors $U$ and $X$ are deterministic, invertible
  functions of each other. Furthermore, the collection of pairs 
  $(X_z,Y_z)_{z\in\{0,1\}^m}$ is independent.
  Applying these observations and the chain rule,
  \begin{align*}
    \sum_{A\subseteq[m]}H_A
    &=
    \sum_{A\subseteq[m]}H(U_A\mid U_{<A},Y)
    =H(U\mid Y)
    =H(X\mid Y)
    =H(X_z\mid Y_z)\cdot n\\
    &=H(X_z\mid X_z+E_z)\cdot n
    =H(E_z)\cdot n
    =h(p)\cdot n\;.\qedhere
  \end{align*}
\end{proof}

\begin{fact}\label{fac:bhattacharyya-properties}
  Let $U$ be uniform in $\{0,1\}$ and $X$, $Y$ be discrete
  random variables. We have:
  \begin{enumerate}
  \item $Z(U\mid XY)\le Z(U\mid X)$.
  \item If $X\in\mathcal{X}$ and $f:\mathcal{X}\to\mathcal{Y}$ is
    injective on the support of $X$, then $Z(U\mid X)=Z(U\mid f(X))$.
  \item If $(U,X)$ is independent of $Y$, then
    $Z(U\mid XY)=Z(U\mid X)$.
  \end{enumerate}
\end{fact}

We omit the proof of Fact~\ref{fac:bhattacharyya-properties},
but all these basic properties are established by direct computations
using~\eqref{eq:14}. For more on the
Bhattacharyya parameter in the context of polar codes,
see, e.g.,~\cite{Arikan09}.

We also state a property which follows by checking
both cases in Definition~\ref{def:decoding-order}:
\begin{fact}\label{fac:order-properties}
  For $A,B\subseteq[m]$:
  \begin{enumerate}
  \item If $A<B$ and $b\in B$, then $A\setminus\{b\}<B\setminus\{b\}$.
  \item If $A<B$ and $a\notin A$, then $A\cup\{a\}<B\cup\{a\}$.
  \end{enumerate}
\end{fact}
Finally, we make use of a standard CLT approximation of $\binom{m}{\le r}$:
\begin{fact}\label{fac:binom-clt}
  Let $r=r(m)$ be such that $r=\frac{m}{2}+\alpha\sqrt{m}+o(\sqrt{m})$.
  Then, we have
  \begin{align*}
    \lim_{m\to\infty}\frac{1}{n}\binom{m}{\le r} = \Phi(2\alpha)\;.
  \end{align*}
  Equivalently, if $r$ is the smallest integer such that
  $\binom{m}{\le r}\ge Rn$ for some $R>0$, then
  $r=\frac{m}{2}+\alpha\sqrt{m}+o(\sqrt{m})$ for $\alpha:=\Phi^{-1}(R)/2$.
\end{fact}

\section{Our Result}\label{sec:theorem}

In our main result we prove that for a binary symmetric channel,
a positive rate 
$\delta$-almost Reed--Muller code succeeds
with high probability under successive decoding.
Due to Theorem~\ref{thm:decoding-bound},
to create such a code it makes sense to delete vectors with largest Bhattacharyya $Z_A$ values:

\begin{definition}\label{def:almost-rm}
  For $0\le r\le m$ and $0\le\delta\le 1$, we fix $\RM(m,r,\delta)$ 
  to be any code
  $\RM(m,\mathcal{A})$ such that:
  \begin{itemize}
  \item $\mathcal{A}\subseteq\binom{m}{\le r}$.
  \item $|\mathcal{A}|=\lceil(1-\delta)\binom{m}{\le r}\rceil$.
  \item For all $A\in\mathcal{A}$ and
    $B\in\binom{m}{\le r}\setminus\mathcal{A}$, we have
     $Z_{A}\le Z_{B}$.
  \end{itemize}
\end{definition}

In particular, $\RM(m,r,\delta)$ is $\delta$-almost
Reed--Muller and its rate is at least $(1-\delta)R$,
where $R$ is the rate of $\RM(m,r)$.
We can now state our main theorem.
For simplicity, we focus only on ``noisier'' binary symmetric channels with
$h(p)\ge1/2$, i.e., $p\ge h^{-1}(1/2)\approx 0.11$.
Since a code that corrects fraction $p$ of random errors also corrects a fraction $p'<p$ of errors, this is without loss
of generality.

\begin{theorem}\label{thm:main}
  Let $0<p<1/2$ and $\delta>0$ be such that
  $0<1-h(p)-2\delta\le 1/2$.
  Then, there exist $R>0$ and
  $r=r(m)$ such
  that:
  \begin{itemize}
  \item Codes $\RM(m,r,\delta)$ have rate at least $R$.
  \item For every $c\in\mathbb{N}$, there exists $m_0=m_0(p,\delta,c)$
    such that for $m>m_0$ the 
    error probability under successive decoding of 
    $\RM(m,r,\delta)$
    is at most $1/n^c$.
  \end{itemize}
  Furthermore, $R$ can be set to $R:=(1-\delta)R_0$, with
  $R_0$ given as
  \begin{align*}
    R_0
    =R_0(p,\delta)&:=\Phi(2\alpha)\;,\\
    \alpha
    =\alpha(p,\delta)&:=2\gamma-
      \sqrt{\frac{9}{32}\ln(2/\delta^2)}\;,\\
    \gamma=\gamma(p,\delta)&:=
    \frac{\Phi^{-1}\left(1-h(p)-2\delta\right)}{2}\;.
  \end{align*}
\end{theorem}


We believe the main interest of this result lies in the qualitative statement:
For every $p<1/2$ and $\delta>0$, there exists $r$ corresponding to a constant rate $R$
such that the successive decoding of $\RM(m,r,\delta)$ corrects fraction $p$ of random errors with high probability. 
In any case, we have an estimate
\begin{align}\label{eq:41}
    R=\delta^{9/8+o(1)}\;,
\end{align}
where $o(1)$ is a function that,
for any fixed $p$, goes
to $0$ as $\delta$ goes to 0.
The derivation of~\eqref{eq:41} is provided in 
Section~\ref{sec:rate-approximation}.

While these rates are much smaller compared to $R=(1-h(p))\delta$
obtainable from~\cite{Has13} or~\cite{AY19}
for $\delta\ge h(p)$ (cf.~Section~\ref{sec:related}),
they hold for values of $\delta$ arbitrarily close to zero.

\section{Proof Outline}\label{sec:outline}

Our strategy for proving Theorem~\ref{thm:main} focuses on inequalities
between values of $Z_A$ for different sets $A$. In particular, as explained
in~\cite{AY19}, if, for a given symmetric channel, we could prove\footnote{
  Throughout the paper we discuss and establish
  inequalities between Bhattacharyya parameters $Z_A$, but our technique
  uses only basic properties listed in Fact~\ref{fac:bhattacharyya-properties}.
  Hence, it is applicable to any measure of information satisfying those
  properties, including conditional entropy $H_A$ and MAP decoding
  error.
} that 
\begin{align}\label{eq:44}
|A|>|B|\implies Z_A\ge Z_B
\end{align}
it would follow that the twin-RM code for that channel and RM code are
equal and, since twin-RM codes achieve capacity,
that Reed--Muller codes achieve capacity on that channel.
Instead, we rely on a weaker property
\begin{align}\label{eq:43}
    B=A\cup\{b\}\setminus\{a,a'\}\implies Z_A\ge Z_B
\end{align}
for $a,a'\in A$.

\subsection{Warm-up: 
\texorpdfstring{$Z_A\ge Z_B$}{Z\_A >= Z\_B} 
for 
\texorpdfstring{$m=4$}{m=4}}\label{sec:warm}

Before presenting our general approach, let us consider the case $m=4$. The main ideas required to 
establish~\eqref{eq:43} can be observed here. 
In this small case, we can actually show~\eqref{eq:44}.



We analyze the task of decoding the message $U$ from $Y$. We have an under-determined linear system of equations $Y=UM+E=[U,E] \cdot [M ; I]$. In this notation $[U,E]$ (the unknowns vector) 
is a concatenation of vectors and $[M ; I]$ 
(the coefficient matrix) is the RM matrix $M$ with the 
identity matrix $I$ underneath it.  


We are interested in the Bhattacharyya parameter $Z_A$, 
so let us focus on the process of decoding
information bit $U_A$ given the preceding\footnote{
Technically, we should
always assume $U_B=0$ for $B<A$, $B\notin\mathcal{A}$.
However, it is known~\cite{Arikan09} that for any fixed value
of $u_{<A}$ it holds that
$Z(U_A\mid U_{<A},Y)=Z(U_A\mid U_{<A}=u_{<A},Y)$
(and the same holds for other statistics like the entropy).
Therefore, we do not need to worry about remembering
that $U_B=0$ for $B\notin\mathcal{A}$.
} 
bits $U_{<A}$ and
the noisy codeword $Y$.
This means that in the system $Y=UM+E$ we can substitute all
values $U_B$ for sets $B<A$. 
Therefore, the only remaining unknowns are 
$U_B$  for $B\ge A$ and $E_z$ for all $z\in\{0,1\}^m$.  


We consider all linear combinations of these equations, dividing them into three types:

\begin{enumerate}
\item Some $U_B$ appears in the equation for $B > A$. 

\item No coordinate of $U$ appears in the equation. This tells us the exact value of $E_{i_1}+...+E_{i_t}$, a sum of a subset of components of $E$ (the components of $E$ that appear in the equation).

\item Out of $U$, only the coordinate $U_A$ appears in the equation. Ideally, we want $U_A$ to appear alone because then we would know its exact value. In general, it will be accompanied
by a sum of error terms $E_{i_1}+...+E_{i_k}$. 
\end{enumerate}

Intuitively, all the information useful for decoding $U_A$ is contained in equations of types 2 and 3. Lemma \ref{prop:entropies} makes this intuition precise. 
The set of all equations of the second type can
be thought of as a vector subspace $\mathcal{H}_{A}$ of $\mathbb{F}_2^{\{0,1\}^m}$.
This holds since for each equation we can think of it
as a binary vector with ones in positions indexed by
variables that occur in the equation.
Furthermore, the difference between two equations of the third type is an equation of the second type. 
Therefore, if we also treat
the set of all equations of the third type as a subset
of $\mathbb{F}_2^{\{0,1\}^m}$ (ignoring the variable $U_A$),
this set is a coset (an affine space) of $\mathcal{H}_{A}$ in the vector space.

To summarize all of the above, the information regarding $U_A$ is represented by a coset of $\mathbb{F}_2^{\{0,1\}^m}$.
This affine subspace is given as $W_A+\mathcal{H}_A$,
where $W_A$ is a translation vector corresponding
to an equation of the third type. 
With some thought, it can be seen that this gives a natural criterion for comparing $Z_A$ and $Z_B$. If we show that the coset of set $A$ is contained in the coset of set $B$,
this means that information available to us when decoding
$U_A$ is a subset of the information available
for $U_B$, and therefore by Fact~\ref{fac:bhattacharyya-properties}.1 we have
$Z_A \geq Z_B$. 


Since the vectors $W_A$ can be interpreted as
elements of $\mathbb{F}_2^{\{0,1\}^m}$, we can
also think of them as evaluation vectors of functions
from $\mathbb{F}_2^m$ to $\mathbb{F}_2$.
Hence, each of them can be identified with
such a function, or in other words with a 
polynomial on $m$ variables.
More so, it turns out that (after a permutation
of coordinates, see Section~\ref{sec:algebraic-outline}
for details) each vector $W_A$ becomes an
evaluation vector of the monomial
$x_{\overline{A}}:=\prod_{a\notin A} x_a$
(this follows from the self-duality of RM codes).
Furthermore, we have that the subspace $\mathcal{H}_A$
is spanned by all previous vectors $(W_B)_{B<A}$. Therefore,
denoting our coset as $\mathcal{C}_A$, it can be written as
\begin{align}\label{eq:45}
    \mathcal{C}_A=
    x_{\overline{A}}+\mathcal{H}_A=
    x_{\overline{A}}+\Span\{x_{\overline{B}}\}_{B<A}\;.
\end{align}

Let us now focus back on $m=4$ and
recall the decoding order in this case:
\begin{align*}
    1234<123<124<134<234<12<13<23<14<24<34<1<2<3<4<\emptyset\;.
\end{align*}
Writing down the cosets in this example, we have
\begin{align}
&
\mathcal{C}_{1234}= x_{\emptyset}+\{0\}\;,
\qquad
\mathcal{C}_{123}=x_{4}+\Span\{x_{\emptyset}\}\;,
\qquad
\mathcal{C}_{124}=x_{3}+\Span\{x_{\emptyset},x_4\}\;,
\qquad\nonumber\\
&
\mathcal{C}_{134}=x_{2}+\Span\{x_{\emptyset},x_4,x_3\}\;,
\qquad
\mathcal{C}_{234}=x_{1}+\Span\{x_{\emptyset},x_4,x_3,x_2\}\;,
\qquad\nonumber\\
&
\mathcal{C}_{12}=x_{34}+\Span\{x_{\emptyset},x_4,x_3,x_2,x_1\}\;,
\qquad\ldots\nonumber
\end{align}



We are hoping to establish inequalities between
the $Z_A$ values by analyzing 
a pure algebraic question of comparing 
the affine spaces $\mathcal{C}_A$. 
As we said, $\mathcal{C}_A\subseteq \mathcal{C}_B$
would imply $Z_A\ge Z_B$.
However, from~\eqref{eq:45} it follows that
if $A<B$, then, on the one hand,
$x_{\overline{A}}\in\mathcal{C}_A$, but on the
other hand, $x_{\overline{A}}\in\mathcal{H}_B$
and therefore $x_{\overline{A}}\notin\mathcal{C}_B$.
Hence, $\mathcal{C}_A\subseteq\mathcal{C}_B$ is
never the case if $A<B$.
To bypass this issue, we use the fact that,
since the components of the vector $E$ represent iid noise,
permuting its indices does not change the underlying probability
distribution. 
If we think of our vectors as polynomials,
each such permutation $\tau$ of $\{0,1\}^m$ induces a linear isomorphism of the polynomial space, 
and the effects of this isomorphism on the polynomials
can be written as a change of variables. 
Therefore, it can be checked that
also $\tau(\mathcal{C}_A)\subseteq\mathcal{C}_B$ ensures
that $Z_A\ge Z_B$ holds.
We illustrate this idea on two examples, corresponding to
two types of permutations that we use throughout this paper.

\begin{enumerate}
    \item If we start with
    \[
    \mathcal{C}_{123}=x_{4}+\Span\{x_{\emptyset}\}
    \]
    and apply the transposition of
    coordinates $x_3\leftrightarrow x_4$,
    we get
    \[
    \tau(\mathcal{C}_{123}) = \tau(x_4+\Span\{x_\emptyset\})
    =x_3+\Span\{x_{\emptyset}\}
    \subseteq x_3+\Span\{x_\emptyset,x_4\}
    =\mathcal{C}_{124}\;.
    \]
    Therefore, we established 
    $\tau(\mathcal{C}_{123})\subseteq\mathcal{C}_{124}$
    and $Z_{123}\ge Z_{124}$.
    
    \item Starting with 
    $\mathcal{C}_{234}=x_1+\Span\{x_\emptyset,x_4,x_3,x_2\}$
    and applying the permutation that maps
    $x_1\to x_1+x_{34}$ and leaves other coordinates
    unchanged, we get
    \[
    \tau(\mathcal{C}_{234})=
    x_1+x_{34}+\Span\{x_{\emptyset},x_4,x_3,x_2\}
    \subseteq x_{34}+\Span\{x_\emptyset,x_4,x_3,x_2,x_1\}
    =\mathcal{C}_{12}\;,
    \]
    establishing $Z_{234}\ge Z_{12}$.
\end{enumerate}

These two examples correspond to two 
types of permutations that we use throughout:
First,
transpositions $x_b\leftrightarrow x_a$   
will give inequalities $Z_A\ge Z_{A\setminus\{a\}\cup\{b\}}$
for $a<b$.
Second, permutations $x_b\to x_b+x_{a,a'}$ will give
$Z_A\ge Z_{A\setminus\{a,a'\}\cup\{b\}}$.
In general, they correspond to Rule~1 and Rule~2 from 
Definition~\ref{def:constructible}.

In our toy case $m=4$, all other inequalities 
$Z_A\ge Z_B$ for $|A|>|B|$ follow in a similar 
way.\footnote{Except for the special cases where $A=[4]$ or $B=\emptyset$, where there is no permutation with
$\tau(\mathcal{C}_A)\subseteq \mathcal{C}_B$. It is not hard to prove these by another argument.}
For $m=5$, this approach proves all $Z_A\ge Z_B$
for $|A|>|B|$ except for a single case of
$Z_{345}$  vs.~$Z_{12}$. For larger $m$, we get more and more
cases not covered by our rules, requiring us to
resort to additional techniques.

\subsection{Proof sketch}

In this section we present the main ingredients in the proof
of Theorem~\ref{thm:main}.
Since we are unable to prove~\eqref{eq:44}, we end up with
Theorem~\ref{thm:main} as a consequence of a weaker set of inequalities. Using our operations in a manner that we
sketched in the previous section, we can show
$Z_A\ge Z_B$ at least for \emph{some} sets with $|A|=|B|+1$. This can be expanded inductively into inequalities with
a larger gap between the sizes of $A$ and $B$.  Ultimately, we take $r=m/2+O(\sqrt{m})$ and
some $k=O(\sqrt{m})$ and show $Z_A\ge Z_B$ 
for \emph{almost every} set $A$ of size $r+k$
and almost every set $B$ of size $r$.

As a consequence of this, imagine that a 
relatively small fraction of sets $B$ with
$|B|\le r$ has non-negligible $Z_B$ values.
Using our relations,
it will turn out that \emph{almost all} sets
$A$ with $|A|\ge r+k$ also have non-negligible $Z_A$ values.
By Theorem~\ref{thm:rm-polarization}, almost all of those sets $A$ must
in fact have $H_A$ close to 1. Ultimately, $r$ and $k$ are chosen
so that we obtain a contradiction with Fact~\ref{fac:entropy-sum}.
The final conclusion is that only a very small fraction of sets $B$ with
$|B|\le r$ can have non-negligible $Z_B$ values. By deleting basis codewords
corresponding to those sets, by Theorem~\ref{thm:decoding-bound}
we obtain a $\delta$-almost Reed--Muller
code that is amenable to successive decoding.

\medskip

Let us expand on this explanation, starting with a general framework
for passing between sets of size $r+k$ and $r$ as described
in the previous paragraph.
We work with orderings on the subsets of $[m]$ (where it
is a good idea to think about them as suborders
of the decoding order from Definition~\ref{def:decoding-order}):

\begin{definition}\label{def:information-consistent}
  We say that a partial order $\ll$ on $\mathcal{P}(m)$ is 
  \emph{information-consistent} if
  \begin{align*}
      A\ll B\implies Z_A\ge Z_B
  \end{align*}
  for all sets $A,B\subseteq[m]$.
\end{definition}
Note that in general whether an order is information-consistent might depend on
the channel (i.e., on $p$). The specific order we introduce later is information-consistent for every $p$.

\begin{definition}\label{def:expansion}
  Let $\ll$ be a partial order on $\mathcal{P}(m)$. We say that $\ll$ is
  \emph{$(\delta,r,k)$-expanding} if, for every collection 
  $\mathcal{B}\subseteq\binom{m}{r}$ of subsets of size $r$,
  letting
  \begin{align}\label{eq:36}
      \mathcal{A}:=\mathcal{A}(\mathcal{B})=\left\{
        A: |A|=r+k\land \exists B\in\mathcal{B}:A\ll B
      \right\}\;,
  \end{align}
  we have
  \begin{align*}
    |\mathcal{B}|\ge\delta\binom{m}{r}\implies
    |\mathcal{A}|\ge(1-\delta)\binom{m}{r+k}\;.
  \end{align*}
\end{definition}
Considering the bipartite graph with one group of vertices being sets of size $r$, the other group sets of size $r+k$ and edges according to the relation $A\ll B$, Definition~\ref{def:expansion} states
a strong expansion property of this graph.
The following lemma makes use of 
Definitions~\ref{def:information-consistent} 
and~\ref{def:expansion} to formalize our strategy
for this part of the proof:
\begin{lemma}\label{lem:expansion}
  Let $0<p<1/2$, 
  let $\ll$ be an information-consistent family of partial orders and let 
  $\delta>0$, 
  $r=r(m)$, $k=k(m)$ satisfy
  \begin{align}\label{eq:35}
    \liminf_{m\to\infty}
    \frac{1}{n}\binom{m}{\ge(r+k)}
    >\frac{h(p)}{(1-\delta)}\;.
  \end{align}
  If $\ll$ is $(\delta,r,k)$-expanding,
  then, for every $c\in\mathbb{N}$, there exists $m_0$ such that for
  $m>m_0$ the error probability of successive decoding of  $\RM(m,r,\delta)$ is less than $1/n^c$.
\end{lemma}

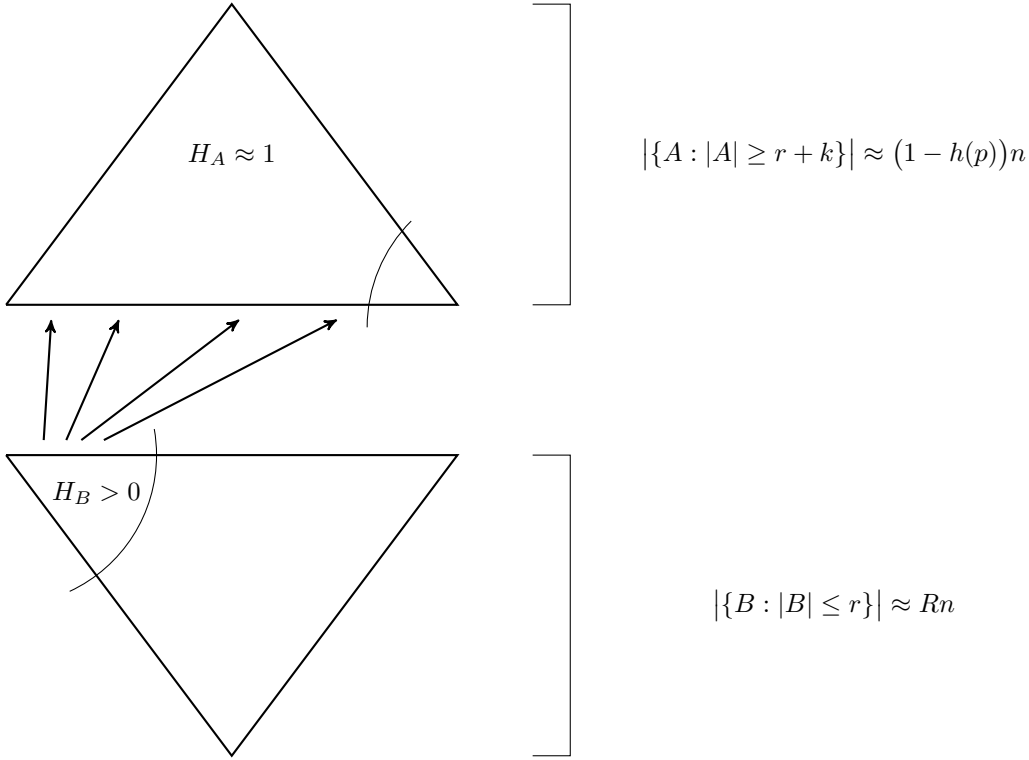
\begin{figure}[ht]\centering

\begin{tikzpicture}
\draw [thick] (-3, 0) -- (3, 0) -- (0, 4) -- (-3, 0);
\draw [thick] (-3, -2) -- (3, -2) -- (0, -6) -- (-3, -2);

\node at (0, 2) {$H_A\approx 1$};
\draw (1.8, -0.3) arc (180:135:2);

\node at (-1.8, -2.5) {$H_B>0$};
\draw (-1, -2) arc (0:10:2); 
\draw (-1, -2) arc (0:-65:2); 

\draw [thick, ->] (-2.5, -1.8) -- (-2.4, -0.2);
\draw [thick, ->] (-2.2, -1.8) -- (-1.5, -0.2);
\draw [thick, ->] (-2, -1.8) -- (0.1, -0.2);
\draw [thick, ->] (-1.7, -1.8) -- (1.4, -0.2);

\draw (4, 0) -- (4.5, 0) -- (4.5, 4) -- (4, 4);
\draw (4, -2) -- (4.5, -2) -- (4.5, -6) -- (4, -6);

\node [align=center] at (8, 2) 
{$\big|\{A:|A|\ge r+k\}\big|\approx\big(1-h(p)\big)n$};
\node [align=center] at (8, -4) 
{$\big|\{B:|B|\le r\}\big|\approx Rn$};
\end{tikzpicture}
\caption{An illustration of the expansion property
from Lemma~\ref{lem:expansion}. For the RM code 
$\RM(m,r)$ with rate $R$, the $(\delta,r,k)$-expansion
property together with RM code polarization 
imply that even a relatively small number
of sets of size at most $r$ with $H_B$ significantly
larger than 0
induces a very large number of sets of size at least $r+k$ with $H_A\approx 1$.
This is in contradiction with $\sum_A H_A=(1-h(p))n$.}
\label{fig:expander}
\end{figure}

In this paper, we take $\ll$ to be a specific order
created by taking the transitive closure of two ``rules'':
\begin{definition}\label{def:constructible}
  Let $A, B\subseteq [m]$. We say that $B$ was obtained from $A$
  by application of Rule~1 if there exist $a<b,a\in A,b\notin A$
  such that $B=A\setminus\{a\}\cup\{b\}$.

  We say that $B$ was obtained from $A$ by application of Rule 2
  if there exist $a,a'\in A,b\notin A\setminus\{a,a'\}$ such that
  $B=A\setminus\{a,a'\}\cup\{b\}$.

  For sets $A,B\subseteq [m]$ we say that
  \emph{$B$ can be constructed from $A$} and
  write $A\ll B$ if
  $B$ can be obtained from $A$ by a finite 
  number of applications of Rules~1 and~2.
\end{definition}

The bulk of our argument consists of proving that the ``can be constructed from''
$\ll$ relation satisfies
the assumptions of Lemma~\ref{lem:expansion}:
\begin{lemma}\label{lem:information-consistent}
  The $\ll$ relation is information-consistent for
  every $0<p<1/2$.
\end{lemma}

We presented
the most important ideas used to prove Lemma~\ref{lem:information-consistent} in Section~\ref{sec:warm}.
We then show that the $\ll$ relation indeed is
$(\delta,k,r)$-expanding for an appropriate choice of parameters:

\begin{lemma}\label{lem:constructible-expansion}
  Let $r=m/2+\alpha\sqrt{m}$ and 
  $k=\beta\sqrt{m}$ such that
  $|\alpha|+\beta\le m^{1/12}$ and $\beta\ge\max(\alpha, -\alpha/2)$.
  Then, the relation $\ll$ is $(\delta, r, k)$-expanding
  for
  \begin{align*}
    \delta=\sqrt{2}\exp\left(-\frac{8}{9}(\beta-\alpha)(2\beta+\alpha)
    +\frac{C}{m^{1/4}}\right)      
  \end{align*}
  for some universal constant $C>0$.
\end{lemma}

Finally, let us say a few words about proving 
Lemma~\ref{lem:constructible-expansion}.
An important intermediate step in its proof is a
sufficient condition for $A\ll B$:
\begin{restatable}{lemma}{dominance}
\label{cor:dominance-without-k}
  Let $A=\{a_1,\ldots,a_{r+k}\}, B=\{b_1,\ldots,b_{r}\}$, with
  $a_1<\ldots<a_{r+k},b_1<\ldots<b_{r}$. Then,
  \begin{align}\label{eq:46}
    \forall 1\le i\le r-k:a_i\le b_{i+k}\implies A\ll B\;.
  \end{align}
\end{restatable}

With the benefit of Lemma~\ref{cor:dominance-without-k},
consider sampling a uniform set $A$ of a given size
$r$ as a standard $m$-step
$\{\pm 1\}$ random
walk conditioned on the endpoint $S=2r-m$.
If we take two random sets $A$, $B$ with
$|A|=r+k$, $|B|=k$,
it can be shown that~\eqref{eq:46}
is implied by the event 
that maximum deviations of respective random walks for both
$A$ and $B$ were not too large. We can then
compute exact tails of these deviations using standard symmetry
arguments and conclude that, in fact, for a \emph{random} choice
of sets $A$ and $B$ relation $A\ll B$ holds
with probability
$1-\delta^2$ in certain range of $r$ and $k$.

In terms of the bipartite graph between sets of size $r+k$ and $r$, 
it means that it is a full bipartite graph except for $\delta^2$ 
fraction of the edges. But we will show in Lemma~\ref{cor:all-sets}
that such a graph must be $(\delta,r,k)$-expanding, which implies
Lemma~\ref{lem:constructible-expansion}. Essentially,
the property of $(\delta,r,k)$-expansion
follows since the full bipartite graph where two vertex sets 
$\mathcal{A}\subseteq\binom{m}{r+k}$ and 
$\mathcal{B}\subseteq\binom{m}{r}$ of density $\delta$ each are designated,
and all edges between $\mathcal{A}$ and $\mathcal{B}$ deleted,
is an extremal example: It minimizes our notion of expansion
among the graphs with $1-\delta^2$ fraction of edges.

\medskip

We note that a substantial improvement to our
$\ll$ relation could lead to a better
result (in terms of larger rates or smaller $\delta$)
via Lemma~\ref{lem:expansion} or its variant.
We also remark that while the rest of the paper
concerns only binary symmetric channel,
it can be checked that Lemma~\ref{lem:expansion}
holds for any binary memoryless symmetric channel
with $(1-h(p))$ substituted by the respective channel capacity.

\section{Proof of Theorem~\ref{thm:main}}

Our proof can be divided into several parts,
corresponding to the lemmas stated in 
Section~\ref{sec:outline}.
We start with Lemma~\ref{lem:information-consistent},
and then move on to Lemma~\ref{lem:constructible-expansion}.
Finally, we prove Lemma~\ref{lem:expansion}
and put together the proof
of Theorem~\ref{thm:main}.

\subsection{Proof of Lemma~\ref{lem:information-consistent}}
\label{sec:algebraic-outline}

For a start, the fact that $\ll$ is a partial order
is easy to see from Definition~\ref{def:constructible}.
What remains is the following property:
\begin{lemma}\label{thm:domination-entropy}
  If $B$ can be constructed from $A$, then $Z_A\ge Z_B$.
\end{lemma}

There are two observations underlying the proof of
Lemma~\ref{thm:domination-entropy}, both already
used in~\cite{AY19} and earlier works.
The first one utilizes 
the algebraic structure of the BSC to simplify the
expression for $Z_A$. Recall the
Ber$(p)$ random vector $E=(E_z)$ and the Reed--Muller matrix $M=(M_{A,z})$:
\begin{lemma}\label{prop:entropies}
  Let $W'=(W'_A)_{A\subseteq [m]}$ be the random vector given by
  \begin{align*}
    W':=EM^{-1}\;.
  \end{align*}
  Then, we have
  \begin{align}\label{eq:15}
    Z_A=Z(U_A\mid U_{<A},Y)=Z(U_A\mid U_A+W'_A,(W'_B)_{B<A})\;.
  \end{align}
\end{lemma}

\begin{proof}Recalling~\eqref{eq:16} and~\eqref{eq:17} and repeatedly
  applying Facts~\ref{fac:bhattacharyya-properties}.2
  and~\ref{fac:bhattacharyya-properties}.3,
  \begin{align*}
    Z_A
    &=Z(U_A\mid U_{<A},Y)
      =Z(U_A\mid U_{<A},YM^{-1})
      =Z(U_A\mid U_{<A},U+EM^{-1})\\
    &=Z\big(U_A\mid U_A+W'_A, U_{<A}, (U_B+W'_B)_{B<A}, (U_B+W'_B)_{A<B}\big)\\
    &=Z\big(U_A\mid U_A+W'_A,U_{<A},(U_B+W'_B)_{B<A}\big)\\
    &=Z(U_A\mid U_A+W'_A, (W'_B)_{B<A}, U_{<A})=Z(U_A\mid U_A+W'_A,(W'_B)_{B<A})
      \;.\qedhere
  \end{align*}
\end{proof}

The second observation is a formalization of a simple fact that relabeling
random variables $E_z$ does not change the right-hand side value in~\eqref{eq:15}:
\begin{fact}\label{fac:permutation}
  Let $\tau:\{0,1\}^m\to\{0,1\}^m$ be a permutation and let $P$ be the relevant
  permutation matrix given by
  \begin{align}\label{eq:18}
    P_{z,z'}=1 \iff z=\tau(z')\;.
  \end{align}
  Furthermore, let $v_1,\ldots,v_k,\tilde{v}_1,\ldots,
  \tilde{v}_{k'}\in\mathbb{F}_2^{\{0,1\}^m}$. 
  Letting $U:=U_A$,
  \begin{align*}
    \mathcal{W}
    &:=\{U+Ev_i^T: i=1,\ldots,k\}\cup\{E\tilde{v}_i^T:i=1,\ldots,k'\}\;,\\
    \mathcal{\tau W}
    &:=\{U+EPv_i^T: i=1,\ldots,k\}\cup\{EP\tilde{v}_i^T:i=1,\ldots,k'\}\;,
  \end{align*}
  we have
  \begin{align}\label{eq:23}
    Z\big(U\mid \mathcal{W}\big)
    =Z\big(U\mid \tau\mathcal{W}\big) \;.
  \end{align}
\end{fact}
\begin{proof}
  Clear, since $U$ is independent of $E$ and random variables $E_z$
  are i.i.d.
\end{proof}

Given $A\subseteq[m]$, using Fact~\ref{fac:rm-matrix-inverse} we see that
\begin{align*}
  W'_A=\sum_zE_z\prod_{i\notin A}(1-z_i)\;.
\end{align*}
To simplify notation, let us define
\begin{align*}
  W_A:=\sum_{z\in\{0,1\}^m}E_z\prod_{i\notin A}z_i\;,\qquad\qquad
  W_{<A}:=(W_B)_{B<A}\;.
\end{align*}
In Section~\ref{sec:permutations} we apply Fact~\ref{fac:permutation}
to the random vector $W'$ to establish
\begin{corollary}\label{cor:z-formula}
  $Z_A = Z\big(U_A\mid U_A+W_A, W_{<A}\big)$.
\end{corollary}

Note that $U_A$ on the
right-hand side in Corollary~\ref{cor:z-formula} is just a uniform bit
independent of everything else, so we might just as well rename it $U:=U_A$.
With Corollary~\ref{cor:z-formula} at hand, we use an elementary
strategy to establish inequalities $Z_A\ge Z_B$. We
demonstrate this by showing that the information
contained in $U+W_A,W_{<A}$ is a ``subset of'' 
information contained in $U+W_B,W_{<B}$
(technically, we show a type of channel degradation
\cite{MP18}).
Informally, we look for permutations $\tau$ of $\{0,1\}^m$ such that
\begin{align*}
  \Span\big\{U+\tau W_A,\tau W_{<A}\big\}
  \subseteq\Span\big\{U+W_B,W_{<B}\big\}\;.
\end{align*}
Facts~\ref{fac:permutation}
and~\ref{fac:bhattacharyya-properties}.1 can then be used
to conclude $Z_A\ge Z_B$.
It is worth noting (and keeping in mind) 
that a permutation of $\{0,1\}^m$ can be thought of as an action
permuting columns of the Reed--Muller matrix $M$.

More precisely, we identify random vectors $W_A$
with monomials, letting
\begin{align*}
  P_A\in\mathbb{F}_2[Z_1,\ldots,Z_m]\;,\qquad\qquad
  P_A(Z_1,\ldots,Z_m):=\prod_{i\notin A}Z_i\;,
\end{align*}
and their linear combinations with multilinear polynomials in
$\mathbb{F}_2[Z_1,\ldots,Z_m]$. Then, letting also
\begin{align*}
  P_{<A}:=\{P_B: B<A\}\;,\qquad\qquad
  (P\circ\tau)(Z_1,\ldots,Z_m):=P(\tau(Z_1,\ldots,Z_m))\;,
\end{align*}
in Section~\ref{sec:permutations} we prove
(where $P_{<A}\circ\tau=\{P\circ\tau:P\in P_{<A}\}$)
\begin{lemma}\label{cor:poly}
  Let $\tau$ be a permutation on $\{0,1\}^m$ and consider
  $A, B\subseteq [m]$. If
  \begin{enumerate}
  \item $P_A\circ\tau \in P_B+\Span\{P_{<B}\}$; and
  \item $P_{<A}\circ\tau \subseteq\Span\{P_{<B}\}$
  \end{enumerate}
  both hold, then $Z_A\ge Z_B$.
\end{lemma}

Lemma~\ref{cor:poly} provides a template for proving $Z_A\ge Z_B$.
In particular, Lemma~\ref{thm:domination-entropy} follows by induction
from the two immediately following lemmas. More precisely,
Lemma~\ref{lem:z-rule-1} covers Rule~1, and Lemma~\ref{lem:algebraic-perm}
together with\footnote{
  Note that in case $B=A\setminus\{a\}$, $a\in A$, the inequality
  $Z_A\ge Z_B$ follows directly from Lemmas~\ref{lem:z-rule-1}
  and~\ref{lem:algebraic-perm} in most cases: One can use Rule 2
  to delete from $A$ its maximum $a'$ and $a$ and insert minimum element
  $b'$ not in $A$,
  and follow up applying Rule~1 to replace back $b'$ with $a'$.

  This \emph{almost always} works, but we defer to Lemma~\ref{lem:z-containment}
  to avoid cumbersome special cases later on.
} Lemma~\ref{lem:z-containment}
cover Rule~2.
The lemmas are proved now
by choosing appropriate $\tau$.

\begin{lemma}\label{lem:z-rule-1}
  Let $a\in A$, $b\notin A$ for some $a<b$ and let
  $B:=A\setminus\{a\}\cup\{b\}$. Then, $Z_A\ge Z_B$. 
\end{lemma}

\begin{lemma}\label{lem:algebraic-perm}
  Let $A\subseteq [m],a,a'\in A,b\notin A$ and
  $B:=A\setminus\{a,a'\}\cup\{b\}$. Then, $Z_A\ge Z_B$.
\end{lemma}

\subsection{Proofs of Lemmas~\ref{lem:z-rule-1}
and~\ref{lem:algebraic-perm}}

\begin{proof}[Proof of Lemma~\ref{lem:z-rule-1}]
  Recall that $\overline{A}=[m]\setminus A$ and note that we have
  \begin{align*}
    a\notin\overline{A},\quad
    b\in\overline{A},\quad
    \overline{B}=\overline{A}\setminus\{b\}\cup\{a\}\;.
  \end{align*}
  Let $\tau:\{0,1\}^m\to\{0,1\}^m$ be given as
  \begin{align}\label{eq:19}
    \tau(z_1,\ldots,z_m)_i:=\begin{cases}
      z_{b}&\text{if $i=a$,}\\
      z_{a}&\text{if $i=b$,}\\
      z_i&\text{otherwise.}
    \end{cases}
  \end{align}
  Since clearly $\tau$ is a permutation, to conclude that $Z_A\ge Z_B$
  we only need to check that the conditions from Lemma~\ref{cor:poly} apply.
  For a start, indeed we have
  \begin{align*}
    P_A\circ\tau=\prod_{i\in\overline{A}}\tau(Z)_i=\prod_{i\in \overline{B}}Z_i=P_B
    \in P_B+\Span\{P_{<B}\}\;.
  \end{align*}
  As for the second condition, we start by observing that $A<B$ in the
  decoding order. Let us take $P_C\in P_{<A}$ and proceed by case
  analysis:
  \begin{itemize}
  \item If $a,b\in C$ or $a,b\notin C$, then we have
    \begin{align*}
      P_C\circ\tau=P_C\in P_{<A}\subseteq P_{<B}\subseteq\Span\{P_{<B}\}\;.
    \end{align*}
  \item If $a\in C$ and $b\notin C$, then
    \begin{align*}
      P_C\circ\tau=\prod_{i\in \overline{C}}\tau(Z)_i
      =\prod_{i\in \overline{C}\setminus\{b\}\cup\{a\}}Z_i
      =P_{C\setminus\{a\}\cup\{b\}}\;,
    \end{align*}
    and, since $C<A$, $a\in A$, $b\notin C$,
    by Fact~\ref{fac:order-properties} we get
    $C\setminus\{a\}\cup\{b\}<B$ and
    $P_{C\setminus\{a\}\cup\{b\}}\in P_{<B}$.
  \item Similarly, if $a\notin C$ and $b\in C$, then
    $P_C\circ\tau=P_{C\setminus\{b\}\cup\{a\}}$. But now it is enough to observe
    that $C\setminus\{b\}\cup\{a\}<C<A<B$ and therefore
    $P_{C\setminus\{b\}\cup\{a\}}\in P_{<A}\subseteq P_{<B}$.
    \qedhere
  \end{itemize}
\end{proof}

\begin{proof}[Proof of Lemma~\ref{lem:algebraic-perm}]
  This time let $\tau:\{0,1\}^m\to\{0,1\}^m$ to be
  \begin{align}\label{eq:20}
    \tau(z_1,\ldots,z_m)_i:=
    \begin{cases}
      z_b+z_az_{a'}&\text{if $i=b$,}\\
      z_i&\text{otherwise.}
    \end{cases}
  \end{align}
  Again, after checking that $\tau$ is a bijection on
  $\{0,1\}^m$, we verify the conditions
  from Lemma~\ref{cor:poly}.
  Note that since $|A|>|B|$, we have $A<B$ in the decoding order.

  First, we see that
  \begin{align*}
    P_A\circ\tau=\prod_{i\in\overline{A}}\tau(Z)_i
    =\prod_{i\in\overline{A}}Z_i+\prod_{i\in\overline{A}\setminus\{b\}\cup\{a,a'\}}Z_i
    =P_A+P_B\in P_B+\Span\{P_{<B}\}\;.
  \end{align*}
  Next, let $C<A$. We consider two cases. First, if $b\in C$, we have
  \begin{align*}
    P_C\circ\tau=\prod_{i\in\overline{C}}\tau(Z)_i=P_C\in P_{<A}\subseteq P_{<B}\;.
  \end{align*}
  On the other hand, if $b\notin C$, we have
  \begin{align*}
    P_C\circ\tau=\prod_{i\in\overline{C}}Z_i+\prod_{i\in\overline{C}\setminus\{b\}\cup\{a,a'\}}Z_i
    =P_C+P_{C\setminus\{a,a'\}\cup\{b\}}\;,
  \end{align*}
  where we used $Z_a^2=Z_a$ over $\mathbb{F}_2$ in case we already had
  $a\in\overline{C}$ (and similar for $a'$).
  Since $C<A<B$, $a,a'\in A$ and $b\notin C$, by applying
  Fact~\ref{fac:order-properties} three times,
  $C<A$ implies $C\setminus\{a,a'\}\cup\{b\}<B$, and consequently
  \begin{align*}
  &
    P_C,P_{C\setminus\{a,a'\}\cup\{b\}}\in P_{<B}\;,\qquad\qquad
    P_C+P_{C\setminus\{a,a'\}\cup\{b\}}\in\Span\{P_{<B}\}\;.
    \qedhere
  \end{align*}
\end{proof}

\subsection{Properties of the construction order}

Here we develop a sufficient condition for $A\ll B$ which we 
subsequently show holds
for ``typical'' $A$ and $B$ of relevant sizes. 
A precise notion that we are going to use is:
\begin{definition}\label{def:good}
  For $d\in\mathbb{N}$, we call a set $A\subseteq [m]$ \emph{$d$-good} if
  for all $i\in[m]:$
  \begin{align*}
    \big|A\cap[i]\big|\le\frac{i}{2}+d\;.
  \end{align*}
\end{definition}
The idea is that no intersection of a good set $A\cap[i]$ contains significantly
more elements than the number expected based just on the size of $A$ (we will always use
this definition for $|A|\approx m/2$). 
The result we prove in this section is:
\begin{lemma}\label{lem:random-set-dominated}
  Let $A$, $B$ be sets with $|A|=r+k$, $|B|=r$ and $d_1,d_2\ge 0$
  be such that $d_1+d_2\le k$. 
  If $\overline{A}$ is $d_1$-good and $B$ is $d_2$-good, then
  $A\ll B$.
\end{lemma}

In order to prove Lemma~\ref{lem:random-set-dominated}, we start with an alternative
characterization of the $\ll$ relation,
divided into two cases $|A|=|B|$ and $|A|>|B|$.

\begin{lemma}\label{lem:domination-equivalence}\mbox{}
  \begin{enumerate}
  \item Let $|A|=|B|$ with $A=\{a_1,\ldots,a_r\},B=\{b_1,\ldots,b_r\}$,
    $a_1<a_2<\ldots<a_r,b_1<\ldots<b_r$. Then, $A\ll B$
    if and only if $a_i \le b_i$ for every $i$.
  \item Given $A, |A|\ge 2$, let $a<a'$ denote the two largest
    elements of $A$ and $b$ the smallest element of
    $\overline{A}\cup\{a,a'\}$.
    Accordingly, let
    \begin{align}\label{eq:06}
      \widetilde{A}:= A\setminus\{a,a'\}\cup\{b\}
    \end{align}
    and $\widetilde{A}^{(k)}$ to be the result of $k$ consecutive
    applications to set $A$ of the operation defined in~\eqref{eq:06}.

    Let $A$ be such that $|A|=|B|+k$, $|B|\ge 1$. Then, $A\ll B$ if and only
    if $\widetilde{A}^{(k)}\ll B$.
  \end{enumerate}
\end{lemma}

Lemma~\ref{lem:domination-equivalence} is proved in Section~\ref{sec:za-zb}
by elementary case analysis. For us its most important consequence
is a sufficient condition for $A\ll B$ that we already
pointed out in Section~\ref{sec:outline}:
\dominance*

\begin{proof}
  First, if $k\ge r$, then $B$ can be constructed from $A$ by
  applying Rule~2 only. Therefore, assume $k<r$ and
  consider $\widetilde{A}^{(k)}$. 
  By Lemma~\ref{lem:domination-equivalence}.2, we
  only need to establish that $B$ can be constructed
  from $\widetilde{A}^{(k)}$. We do this by checking the condition
  from Lemma~\ref{lem:domination-equivalence}.1.
  
  If
  $\widetilde{A}^{(k)}=\{1,\ldots,r\}$, then
  clearly $\widetilde{A}^{(k)}\ll B$ and we are done.
  Otherwise, we can write $\widetilde{A}^{(k)}=
  \{\tilde{a}_1,\ldots,\tilde{a}_k,a_{1},\ldots,a_{r-k}\}$,
  where $\tilde{a}_1,\ldots,\tilde{a}_k$ are $k$ elements added
  in the applications of Rule~2. 
  Let $a':=\min([m]\setminus\widetilde{A}^{(k)})$ be the
  smallest element not in $\widetilde{A}^{(k)}$ and let
  $c_i$ be the $i$-th smallest element of $\widetilde{A}^{(k)}$.
  We consider two cases. First, if $c_i<a'$, then clearly
  $c_i=i\le b_i$. Second, if $a'<c_i$, then note that also
  $\max\{\ta_1,\ldots,\ta_k\}<c_i$ (in particular $i>k$) and therefore
  $c_i=a_{i-k}$ and $a_{i-k}\le b_i$ holds by assumption.
\end{proof}

The definition of a good set can be linked to
Lemma~\ref{cor:dominance-without-k} by
\begin{fact}\label{fac:good-ai}
  Let $A=\{a_1,\ldots,a_r\},a_1<\ldots<a_r$. If $A$ is $d$-good, then
  for all $i\in[r]$,
  \begin{align}\label{eq:08}
    a_i\ge2i-2d\;.
  \end{align}
  Similarly, if $\overline{A}$ is $d$-good, then
  $a_i\le 2i+2d$.
\end{fact}

\begin{proof}
  To prove~\eqref{eq:08}, assume otherwise,
  i.e., that there exists $i$ such that
  $a_i<2i-2d$ and consequently
  $a_i\le2i-2d-1$.
  This is equivalent to saying that
  (note that we can assume wlog that $2i-2d>1$)
  \begin{align*}
    \Big|A\cap\left\{1,\ldots,2i-2d
    -1\right\}\Big|\ge i
    \;,
  \end{align*}
  however since $A$ is $d$-good, we also have
  \begin{align*}
    \Big|A\cap\left\{1,\ldots, 2i-2d-1\right\}\Big|
    \le\frac{2i-2d-1}{2}+d<i\;,
  \end{align*}
  a contradiction.

  As for the second statement, the proof is similar:
  Assuming wlog $2i+2d<m$,
  if there exists
  $i$ with $a_i>2i+2d$, then
  \begin{align*}
    \Big|A\cap\left\{1,\ldots,2i+2d\right\}
    \Big|\le i-1\;,
  \end{align*}
  but since $\overline{A}$
  is $d$-good, we have $|A\cap\{1,\ldots,i\}|\ge i/2-d$
  for every $i$, in particular
  \begin{align*}
    \Big|A\cap\left\{1,\ldots,2i+2d\right\}\Big|
    \ge \frac{2i+2d}{2}-d
    =i\;,
  \end{align*}
  another contradiction.
\end{proof}

Now we are ready to complete the proof of Lemma~\ref{lem:random-set-dominated}:
\begin{proof}[Proof of Lemma~\ref{lem:random-set-dominated}]
  Assume that 
  $\overline{A}$ is $d_1$-good and $B$ is $d_2$-good. Then, by
  Fact~\ref{fac:good-ai}, $a_i\le 2i+2d_1$ and $b_i\ge 2i-2d_2$.
  In particular, for $1\le i\le r-k$,
  \begin{align*}
    a_i\le2i+2d_1\le2(i+k)-2d_2\le b_{i+k}\;,
  \end{align*}
  and, by Lemma~\ref{cor:dominance-without-k}, $A\ll B$.
\end{proof}

\subsection{Proof of Lemma~\ref{lem:constructible-expansion}}

As we indicated, our strategy to obtain
Lemma~\ref{lem:constructible-expansion}
is to establish that,
for $r=m/2+\alpha\sqrt{m}$ and
$k=\beta\sqrt{m}$ for $\beta>0$ large enough compared to $|\alpha|$,
for almost every pair of sets
$(A,B)$ with $A$ of size $r+k$ and $B$ of size $r$, $B$ can be constructed
from $A$. 
The precise form of this statement is:

\begin{lemma}\label{cor:random-set-dominated}
  Let $r=m/2+\alpha\sqrt{m}$ and $k=\beta\sqrt{m}$ such that
  $|\alpha|+\beta\le m^{1/12}$ and
  $\beta\ge\max(\alpha,-\alpha/2)$.
  Consider a random choice of two independent uniform sets $(A, B)$
  conditioned on $|A|=r+k$ and $|B|=r$. Then, we have
  \begin{align*}
    \Pr[\lnot(A\ll B)]\le
    2\exp\Big(-\frac{16}{9}(\beta-\alpha)(2\beta+\alpha)+
    \frac{C}{m^{1/4}}\Big)
  \end{align*}
  for some universal $C>0$.
\end{lemma}

Lemma~\ref{cor:random-set-dominated} is proved by showing that
typical sets of sizes $r+k$ and $r$ are likely to be $d$-good for
appropriate $d$, and therefore susceptible to applying
Lemma~\ref{lem:random-set-dominated}.
In order to do that, we need a variant of a classic tail bound on the
maximum of a simple integer random walk
(for some background on this technique see,
e.g., Chapter~III in~\cite{Fel68}):

\begin{lemma}\label{thm:maximum-tail}
  Let $X_1,X_2,\ldots,X_m$ be a uniform i.i.d.~sequence with 
  $X_i\in\{-1,1\}$.
  Letting $S_i:=\sum_{j=1}^iX_j$, $M:=\max_{0\le i\le m}S_i$ 
  we have, for any $s,d\in\mathbb{Z}$
  such that $d\ge\max(s,0)$,
  $r:=(m+s)/2\in\mathbb{Z}$
  and $0\le r\le m$,
  \begin{align}\label{eq:07}
    \Pr[M\ge d\mid S_m=s] = \frac{\binom{m}{r-d}}{\binom{m}{r}}\;.
  \end{align}
\end{lemma}
For example, for $m=2k$ and $s=0$ we get
$\Pr[M\ge d\mid S_m=0]=\binom{2k}{k-d}/\binom{2k}{k}$.

\begin{proof}
  Consider a realization of the sequence
  $x=(x_1,\ldots,x_m)\in\{-1,1\}^m$ such that $M(x)\ge d$ and
  $S_m(x) = s$.
  Let $T$ be the largest index such that $S_T = d$ and take
  $y=(y_1,\ldots,y_m)$ to be the mirror image of $x$ after time
  $T$, i.e.,
  \begin{align}\label{eq:47}
    y_i:=\begin{cases}
      x_i&\text{if }i\le T\;,\\
      -x_i&\text{if }i>T\;.
    \end{cases}
  \end{align}
  Note that this operation creates a bijection between sequences
  $x$ such that $M\ge d$ and $S_m=s$ and sequences such that
  $S_m=2d-s$. Indeed, sequence $y$ has 
  $S_m(y)=d-(s-d)=2d-s$. On the other hand, since $s\ge d$,
  we have $2d-s\ge d$ and therefore any walk that ends with
  $S_m=2d-s$ must achieve $S_j=d$ at some point.
  But knowing that it is easy to find the inverse of~\eqref{eq:47},
  establishing the bijection.
  
  Since a priori every walk has the same probability $2^{-m}$,
  we have
  \begin{align*}
    \Pr{[M\ge d\mid S_m=s]}
    &=\frac{\Pr[M\ge d\land S_m=s]}{\Pr[S_m=s]}
      =\frac{\Pr[S_m=2d-s]}{\Pr[S_m=s]}
      =\frac{\binom{m}{r-d}}{\binom{m}{r}}\;.\qedhere
  \end{align*}
\end{proof}


We need to connect Lemma~\ref{thm:maximum-tail} to the notion of
good sets. This is done by
\begin{corollary}\label{cor:good-vs-random-walk}
  Let $A$ be a random uniform set of size $r$ and $d\in\mathbb{N}$
  s.t.~$2d+1\ge s:=2r-m$. Then,
  \begin{align}\label{eq:27}
    \Pr[\text{$A$ not $d$-good}]=
    \frac{\binom{m}{r-2d-1}}{\binom{m}{r}}\;.
  \end{align}
\end{corollary}
\begin{proof}
  Define random variables $X_1,\ldots,X_m$ as
  \begin{align*}    
    X_i=X_i(A):=
    \begin{cases}
      1&\text{if }i\in A,\\
      -1&\text{if }i\notin A,
    \end{cases}
    \qquad
    S_i=S_i(A):=\sum_{j=1}^i X_j\;,
    \qquad      
    M=M(A):=\max_{0\le i\le m}S_i\;.
  \end{align*}
  Recalling the setting of Lemma~\ref{thm:maximum-tail}, note that
  $X_1,\ldots,X_m$ are distributed as an i.i.d.~uniform $\{-1,1\}^m$ sequence
  conditioned on $S_m=2r-m=s$. Furthermore, we have
  \begin{align*}
    S_i = 2|A\cap[i]|-i\;,
  \end{align*}
  and therefore
  \begin{align*}
    \text{$A$ is $d$-good}\iff \forall 1\le i\le m: S_i\le 2d
    \iff
    M(A)\le 2d\;.
  \end{align*}
  Therefore, we can apply
  Lemma~\ref{thm:maximum-tail} and get
  \begin{align*}
    \Pr[\text{$A$ not $d$-good}]
    &=\Pr[M\ge2d+1\mid S_m=s]=
    \frac{\binom{m}{r-2d-1}}{\binom{m}{r}}\;.\qedhere
  \end{align*}
\end{proof}

The proof of Lemma~\ref{cor:random-set-dominated} chooses appropriate
$d_1,d_2=\Theta(\sqrt{m})$ and uses
Lemma~\ref{lem:random-set-dominated} and union bound to show
that
\begin{align*}
  \Pr[\lnot(A\ll B)]\le
  \Pr[\text{$\overline{A}$ not $d_1$-good}]+
  \Pr[\text{$B$ not $d_2$-good}]
\end{align*}
is small by approximating the expression in~\eqref{eq:27}. The details
are provided in Section~\ref{sec:za-zb-details}.

The final ingredient of 
the proof of Lemma~\ref{lem:constructible-expansion} 
consists of connecting 
Lemma~\ref{cor:random-set-dominated} 
with our notion of expansion:

\begin{lemma}\label{cor:all-sets}
  Consider two independent random sets $(A, B)$ 
  such that $|A|=r+k$ and $|B|=r$
  and assume that we have $\Pr[A\ll B]\ge 1-\delta^2$.
  Then, the order $\ll$ is $(\delta, r, k)$-expanding.
\end{lemma}

\begin{proof}
  Let $\mathcal{B}\subseteq\binom{m}{r}$ such that
  $|\mathcal{B}|\ge\delta\binom{m}{r}$ and let
  $\mathcal{A}:=\{A\in\binom{m}{r+k}:\exists B\in\mathcal{B}:A\ll B\}$.
  By definition, $A\ll B$ and $B\in\mathcal{B}$ implies $A\in\mathcal{A}$.
  Hence,
  \begin{align*}
    \Pr[A\in\mathcal{A}]\Pr[B\in\mathcal{B}]
    &=\Pr[A\in\mathcal{A}\land B\in\mathcal{B}]
    \ge\Pr[A\ll B\land B\in\mathcal{B}]\\
    &\ge\Pr[B\in\mathcal{B}]-\Pr[\lnot(A\ll B)]
    \ge\Pr[B\in\mathcal{B}]-\delta^2\;.
  \end{align*}
  Consequently,
  \begin{align*}
    \frac{|\mathcal{A}|}{\binom{m}{r+k}}
    &=
    \Pr[A\in\mathcal{A}]\ge 1-\frac{\delta^2}{\Pr[B\in\mathcal{B}]}\ge
    1-\delta\;.\qedhere
  \end{align*}
\end{proof}

\begin{proof}[Proof of Lemma~\ref{lem:constructible-expansion}]
  Immediately from Lemma~\ref{cor:random-set-dominated}
  and Lemma~\ref{cor:all-sets}.
\end{proof}

\subsection{Proof of Lemma~\ref{lem:expansion}}

Before we prove Lemma~\ref{lem:expansion}, we need a simple fact
stating that in expectation the density of sets with high $Z_A$
increases with the size $r$.

\begin{fact}\label{fac:one-deg-up}
\mbox{}
\begin{enumerate}
    \item Let 
    $\mathcal{B}:=\left\{B\in\binom{m}{\le r}:Z_B\ge\eps\right\}$ and let 
    $\mathcal{B}_r:=\mathcal{B}\cap\binom{m}{r}$.
    Then, $|\mathcal{B}|\ge\delta\binom{m}{\le r}$ implies
    $|\mathcal{B}_r|\ge\delta\binom{m}{r}$.
    \item Similarly, let    $\mathcal{A}:=\left\{A\in\binom{m}{\ge r}:Z_A\ge\eps\right\}$ and  
    $\mathcal{A}_r:=\mathcal{A}\cap\binom{m}{r}$.
    Then, $|\mathcal{A}_r|\ge\delta\binom{m}{r}$ implies
    $|\mathcal{A}|\ge\delta\binom{m}{\ge r}$.
\end{enumerate}
\end{fact}

\begin{proof}
  Let $0\le k<r$ and
  let $A,B$ be random subsets of $[m]$ chosen such that 
  $B$ is uniform among sets of size $k$ 
  and $A=B\cup\{a\}$, where $a$ is
  a uniform element not in $B$. Note that the
  marginal distribution of $A$ is uniform over sets of
  size $k+1$. Furthermore, by Lemma~\ref{lem:z-containment},
  in this random experiment we always have $Z_A\ge Z_B$.
  Therefore, $B\in\mathcal{B}_k$ implies $A\in\mathcal{B}_{k+1}$.
  Letting $\mathcal{B}_k:=\mathcal{B}\cap\binom{m}{k}$,
  we have
  \begin{align}\label{eq:34}
      \frac{|\mathcal{B}_k|}{\binom{m}{k}}=
      \Pr[B\in\mathcal{B}_k]\le
      \Pr[A\in\mathcal{B}_{k+1}]
      =\frac{|\mathcal{B}_{k+1}|}{\binom{m}{k+1}}\;,
  \end{align}
  which clearly implies
  \begin{align*}
    \frac{|\mathcal{B}_r|}{\binom{m}{r}}
    \ge \frac{|\mathcal{B}|}{\binom{m}{\le r}}\;,
  \end{align*}
  establishing the first point. The proof of the second
  point is symmetrical.
\end{proof}

\begin{proof}[Proof of Lemma~\ref{lem:expansion}]
  Let 
  \begin{align*}
  \mathcal{B}:=\left\{B\in\binom{m}{\le r}: 
  Z_B\ge 1/n^{c+1}\right\}\;,
  \qquad\qquad
  \mathcal{A}:=\left\{
  A\in\binom{m}{\ge(r+k)}:
  Z_A\ge 1/n^{c+1}
  \right\}\;.
  \end{align*}
  Our objective is to show that, for $m$ large enough, 
  $|\mathcal{B}|<\delta \binom{m}{\le r}$, since then
  by Theorem~\ref{thm:decoding-bound} we obtain that
  code $\RM(m,r,\delta)$ has successive decoding
  error probability at most $\binom{m}{\le r}\cdot n^{-(c+1)}\le n^{-c}$.
  
  Assume otherwise, i.e., 
  $|\mathcal{B}|\ge\delta \binom{m}{\le r}$.
  By Fact~\ref{fac:one-deg-up}, we have
  $\mathcal{B}\cap\binom{m}{r}
  \ge\delta\binom{m}{r}$.
  Since $\ll$ is both information-consistent and 
  $(\delta,r,k)$-expanding, also
  $\mathcal{A}\cap\binom{m}{r+k}\ge
  (1-\delta)\binom{m}{r+k}$, and, applying the 
  other part of Fact~\ref{fac:one-deg-up},
  \begin{align*}
      |\mathcal{A}|\ge(1-\delta)\binom{m}{\ge(r+k)}\;.
  \end{align*}
  Recall~\eqref{eq:35} and choose $\eps>0$ such that
  \begin{align*}
  \liminf_{m\to\infty}\binom{m}{\ge(r+k)}/n>
  \frac{h(p)}{(1-\eps)(1-\delta)}\;.
  \end{align*}
  Applying
  Theorem~\ref{thm:rm-polarization} for this $\eps$
  and $\xi=1/4$, we have,
  for $m$ large enough,
  \begin{align*}
      \left|\left\{
      A\in\binom{m}{\ge(r+k)}: H_A\ge 1-\eps
      \right\}\right|
      &\ge (1-\delta)\binom{m}{\ge(r+k)}
      -\frac{n}{m^{1/4}}\\
      &>\frac{h(p)}{(1-\eps)}\cdot n\;.
  \end{align*}
  But that implies
  $\sum_{A\subseteq[m]}H_A>h(p)\cdot n$,
  which is in contradiction with
  Fact~\ref{fac:entropy-sum}.
\end{proof}

\subsection{Proof of Theorem~\ref{thm:main}}\label{sec:proof-main}

Fix $p$ and $\delta$.
Recall from the statement of the theorem how the values
of $\alpha, \gamma, R_0$ and $R$ are determined by $p$ and $\delta$.
We choose $r$ to be the smallest number such that the rate of the 
Reed--Muller code $\RM(m,r)$ exceeds $R_0$.
Recalling Definition~\ref{def:almost-rm}, the rate of the code
$\RM(m,r,\delta)$ exceeds $R_0(1-\delta)=R$,
so the actual work lies in showing the successive decoding error
bound. 

To that end, consider the constructible $\ll$ order
from Definition~\ref{def:constructible}. 
By Lemma~\ref{lem:information-consistent}, it is information-consistent.
By Fact~\ref{fac:binom-clt}, we know that
$r=\frac{m}{2}+\alpha\sqrt{m}+o(\sqrt{m})$.
Furthermore, note that since we assumed $1-h(p)-2\delta\le 1/2$,
we have $\gamma\le 0$ and $\alpha<2\gamma\le\gamma$ and
accordingly we can let $\beta:=\gamma-\alpha> 0$ and 
$k:=\lfloor\beta\sqrt{m}\rfloor=\beta\sqrt{m}+O(1)$.

What is more, $\alpha<2\gamma<0$ implies
$\beta>-\alpha/2>0>\alpha$.
Hence, letting $\alpha':=(r-m/2)/\sqrt{m}$ and $\beta':=k/\sqrt{m}$,
we have $\alpha'=\alpha+o(1)$ and $\beta'=\beta+O(1/\sqrt{m})$ and 
\begin{align*}
    |\alpha'|+\beta'\le m^{1/12}\;,\qquad\qquad
    \beta'\ge\max(\alpha',-\alpha'/2)\;,
\end{align*}
so that we can apply Lemma~\ref{lem:constructible-expansion}
and obtain that, for $m$ large enough,
the order $\ll$ is $(\delta', r, k)$-expanding for
\begin{align*}
    \delta'
    &=\sqrt{2}\exp\left(
    -\frac{8}{9}(\beta'-\alpha')(2\beta'+\alpha')+\frac{C}{m^{1/4}}
    \right)
    =\sqrt{2}\exp\left(-\frac{8}{9}(\beta-\alpha)(2\beta+\alpha)\right)
    +o(1)\\
    &=\sqrt{2}\exp\left(-\frac{8}{9}(\gamma-2\alpha)(2\gamma-\alpha)
    \right)+o(1)\;,
\end{align*}
and, since
\begin{align*}
  \frac{1}{2}(\gamma-2\alpha)(2\gamma-\alpha)
  &=\left(\frac{5}{4}\gamma-\alpha\right)^2
    -\frac{9}{16}\gamma^2
    =\left(\sqrt{\frac{9}{32}\ln(2/\delta^2)}-\frac{3}{4}\gamma\right)^2
        -\frac{9}{16}\gamma^2
    > \frac{9}{32}\ln(2/\delta^2)\;,
\end{align*}
also
\begin{align*}
    \delta'
    &=\sqrt{2}\exp\left(-\frac{16}{9}\cdot\frac{1}{2}
    (\gamma-2\alpha)(2\gamma-\alpha)
    \right)+o(1)
    <\sqrt{2}\exp\left(
    -\frac{16}{9}\cdot\frac{9}{32}\ln(2/\delta^2)
    \right)=\delta\;.
\end{align*}
Therefore, for large $m$ we have that the order $\ll$
is $(\delta, r, k)$-expanding.
Finally, we check that
\begin{align*}
    \liminf_{m\to\infty}
    \frac{\binom{m}{\ge(r+k)}}{n}
    =\Phi(-2\gamma)=1-\Phi(2\gamma)
    =h(p)+2\delta>
    \frac{h(p)}{(1-\delta)}\;,
\end{align*}
hence Lemma~\ref{lem:expansion} applies and 
successive decoding of code $\RM(m,r,\delta)$ fails
with probability at most $n^{-c}$.
\qed

\section{Lower Bound}\label{sec:lower-bound}

One might wonder how tight is our analysis of
the expansion properties of the constructible
$\ll$ order. In particular, could we improve upon
Lemma~\ref{lem:constructible-expansion} and get
a constant rate $R$ for codes $\RM(m, r, \delta)$
with $\delta(m)=o(1)$? In this section,
we answer this question in the negative.
We exhibit an assignment of ``possible entropies''
to sets $A\to H(A)$ that respects the condition
$A\ll B\implies H(A)\ge H(B)$, satisfies
$\sum_A H(A)=h(p)n$ and contains a $\delta>0$
fraction of sets with $|A|\le r$ and $H(A)=1$.
If the actual entropies $H_A$ behave similarly,
then $\delta$ fraction of components would have to be
removed from $\RM(m, r)$ to ensure that the successive decoding
corrects random errors.

Therefore, in order to make progress on the capacity
conjecture using polarization theory, either more $H_A\ge H_B$ inequalities need to be proved or new ingredients 
introduced to our approach. We now state
our lower bound:

\begin{theorem}\label{thm:lower-bound}
  Given $0<R<1$, let $r=r(R,m)$ be the smallest $r$ such that
  $\binom{m}{\le r}\ge Rn$. For every $0<R,\eps<1$, 
  there exists $\delta=\delta(\eps,R)>0$, $m_0=m_0(\eps,R)$
  and a family of functions $H=H^{(m)}:\mathcal{P}(m)\to[0,1]$ such that for $m>m_0$:
  \begin{itemize}
  \item $A\ll B$ implies $H(A)\ge H(B)$.
  \item $\left|\big\{A\subseteq[m]: H(A)\notin\{0,1\}\big\}\right|\le 1$.
  \item $\sum_{A\subseteq[m]} H(A)=\eps n$.
  \item $\left|\left\{A\in\binom{m}{\le r}:H(A)=1\right\}\right|
    \ge\delta\binom{m}{\le r}$.
  \end{itemize}
\end{theorem}

Intuitively, the parameter $R$ corresponds to the rate of the code
and $\eps$ to the noise entropy  
(i.e., $1-\text{capacity}$) of a binary input channel.
Theorem~\ref{thm:lower-bound} asserts that it is possible to
``assign entropies'' $H_A$ such that: 1) They are consistent with the $\ll$
relation. 2) In light of polarization (Theorem~\ref{thm:rm-polarization}),
almost all of them are zero or one. 3) The sum of entropies is $\eps n$ 
as required by Fact~\ref{fac:entropy-sum}. 
4) Yet, for any rate $R$ Reed--Muller code
$\RM(m, r)$, a $\delta(\eps,R)>0$ fraction of sets $A$ of size at most
$r$ has $H(A)=1$.

However, $H_A\approx 1$ implies that the bit corresponding to set $A$
cannot be decoded under successive decoding. Consequently, if the entropy
values are given by function $H(\cdot)$, at least $\delta$ fraction of basis
codewords has to be deleted from $\RM(m,r)$ to make successive decoding
work. Therefore, at least for the
purposes of successive decoding, more constraints on function $H(\cdot)$
are needed.

The main conceptual ingredient we need to prove
Theorem~\ref{thm:lower-bound} is the following
easy lemma.
Recall that for $A\subseteq[m]$ we have $S_i(A)=2|A\cap[i]|-i$:
\begin{lemma}\label{lem:run-plus-one}
  Let $\ell:=\lfloor m/2\rfloor$, $0\le r<m$, and $k\in\mathbb{Z}$ and let
  \begin{align*}
    \mathcal{B}_{m,r,k}
    &:=\left\{B\in\binom{m}{r}:
      S_{\ell}(B)>k\right\}\;,\\
    \mathcal{A}_{m,r,k}
    &:=\left\{A\in\binom{m}{r+1}:\exists B\in\mathcal{B}_{m,r,k}\text{ s.t.~}A\ll B
      \right\}\;.
  \end{align*}
  Then, $\mathcal{B}_{m,r+1,k}\subseteq\mathcal{A}_{m,r,k}\subseteq\mathcal{B}_{m,r+1,k-2}$,
  where the left containment holds
  if $\ell+k<2r$.
\end{lemma}
We will use this lemma for $k\le O(\sqrt{m})$ and $r\ge \ell-O(\sqrt{m})$,
so the condition $\ell+k<2r$ will be easily satisfied.
\begin{proof} 
  Fix $m,r,k$ and let $\mathcal{A}:=\mathcal{A}_{m,r,k}$.
  For the first containment, let $A\in\mathcal{B}_{m,r+1,k}$, i.e.,
  $S_{\ell}(A)>k$ and let $B:=A\setminus\{a\}$, where $a$ is the maximum element
  of $A$. Since $A\supseteq B$, we have $A\ll B$. At the same time,
  $S_{\ell}(B)>k$ (note that if $a\le \ell$, then $\ell+k<2r$ implies
  $S_{\ell}(A)=2(r+1)-\ell>k+2$), hence $B\in\mathcal{B}_{m,r,k}$ and $A\in\mathcal{A}$.
  
  For the second containment, let $A\in\mathcal{A}$,
  meaning $A\ll B$ for some $B\in\mathcal{B}_{m,r,k}$.
  Recall from Lemma~\ref{lem:domination-equivalence} that
  $\tA=A\setminus\{a,a'\}\cup\{b\}$, where $a$, $a'$ are two largest elements
  in $A$ and $b$ is the smallest element not in $A\setminus\{a,a'\}$ and that
  $A\ll\tA\ll B$.

  Since clearly $|\tA\cap[\ell]|\le|A\cap[\ell]|+1$, we have
  \begin{align*}
    S_{\ell}(A)=2|A\cap[\ell]|-\ell\ge 2|\tA\cap[\ell]|-\ell-2=S_{\ell}(\tA)-2\;.
  \end{align*}
  But recalling Lemma~\ref{lem:domination-equivalence}.1, it is also clear
  that $S_{\ell}(\tA)\ge S_{\ell}(B)$ and hence
  \begin{align*}
    S_{\ell}(A)\ge S_{\ell}(\tA)-2\ge S_{\ell}(B)-2>k-2\;,
  \end{align*}
  hence $A\in\mathcal{B}_{m,r+1,k-2}$, as claimed.
\end{proof}

Given $R$ and $\eps$, Theorem~\ref{thm:lower-bound} is proved by
setting
\begin{align*}
  \mathcal{B}_0:=\mathcal{B}_{m,r-\gamma\sqrt{m},\gamma'\sqrt{m}}
\end{align*}
for appropriately chosen $\gamma$ and $\gamma'$. Then, we set
$H(A)=1$ if and only if there exists $B\in\mathcal{B}_0$ such that
$A\ll B$. Lemma~\ref{lem:run-plus-one} together with the CLT approximation
from Fact~\ref{fac:binom-clt} allow us to control both
$\sum_{A\subseteq[m]} H(A)$ and $\sum_{A\in\binom{m}{\le r}} H(A)$.
Details are given in Section~\ref{sec:lower-bound-details}.

\section{Remaining Proofs}

\subsection{Approximation from~\texorpdfstring{\eqref{eq:41}}{(9)}}
\label{sec:rate-approximation}
We obtain the estimate stated in~\eqref{eq:41}.
We start with 
\begin{align*}
    \frac{1}{C\max(x,1)}\cdot\exp\left(-\frac{x^2}{2}\right)\le
    \Phi(-x)\le C\cdot\exp\left(-\frac{x^2}{2}\right)\;,
\end{align*}
for every $x\ge 0$ and some universal $C\ge 1$.
This is a standard Gaussian estimate that can be
proved, e.g., by integration by parts.
In particular, substituting in the right part
$y:=C\cdot\exp(-x^2/2)$ and solving 
$x=\sqrt{2\ln(C/y)}$ we get
\begin{align*}
    \Phi^{-1}(y)\ge -\sqrt{2\ln\frac{C}{y}}
\end{align*}
for every $0\le y\le 1$. Now we can calculate, letting $h:=h(p)$,
\begin{align}
  R
  &=
    (1-\delta)\cdot\Phi\left(4\gamma-\sqrt{\frac{9}{8}\ln\frac{2}{\delta^2}}\right)
    =
  (1-\delta)\cdot\Phi\left(
  2\Phi^{-1}(1-h-2\delta)
  -\sqrt{\frac{9}{8}\ln\frac{2}{\delta^2}}
  \right)\nonumber\\
  &\ge (1-\delta)\cdot\Phi\left(-2\sqrt{2\ln\frac{C}{1-h-2\delta}}-
    \sqrt{\frac{9}{8}\ln\frac{2}{\delta^2}}\right)
    \nonumber\\
  &=(1-\delta)\cdot\Phi\left(-\sqrt{\frac{9}{8}\ln\frac{2}{\delta^2}}+O(1)\right)
  =\delta^{9/8+o(1)}\;,\nonumber
\end{align}
where the $o(1)$ function goes to $0$ as $\delta$ goes to $0$ for fixed $p$. 

\subsection{Algebraic column permutations}\label{sec:permutations}

In this section we provide two proofs omitted from
Section~\ref{sec:algebraic-outline}:
Corollary~\ref{cor:z-formula} and
Lemma~\ref{cor:poly}.

\begin{proof}[Proof of Corollary~\ref{cor:z-formula}]
  By Lemma~\ref{prop:entropies}, $Z_A=Z\big(U_A\mid U_A+W'_A,(W'_B)_{B<A}\big)$.
  Define a permutation $\tau$ of $\{0,1\}^m$ as
  \begin{align*}
    \tau(z_1,\ldots,z_m):=(1-z_1,\ldots,1-z_m)
  \end{align*}
  and let $P$ be the respective permutation matrix given by~\eqref{eq:18}.
  Note that matrix $P$ maps the random vector $E=(E_z)_{z\in\{0,1\}^m}$
  to $EP=(E_{\tau(z)})_{z\in\{0,1\}^m}$. Therefore, using
  Fact~\ref{fac:rm-matrix-inverse}, for any set $B$
  \begin{align*}
    (EPM^{-1})_B
    &= \sum_z E_{\tau(z)}\prod_{i\notin B}(1-z_i)
    =\sum_zE_{\tau(z)}\prod_{i\notin B}\tau(z)_i
      =\sum_zE_z\prod_{i\notin B}z_i=W_B\;.
  \end{align*}
  Applying Fact~\ref{fac:permutation} for $\tau$ and $v_1^T,\tilde{v}^T_1,\ldots,
  \tilde{v}^T_{k'}$
  given as columns of $M^{-1}$:
  \begin{align*}
    v_1^T
    &:=(M^{-1})_{(\cdot,A)}\;,\\
    (\tilde{v}_1^T,\ldots,\tilde{v}_{k'}^T)
    &:=\Big(M^{-1}_{(\cdot,B)}\Big)_{B<A}\;,
  \end{align*}
  we obtain
  \begin{align*}
    Z\big(U_A\mid U_A+W'_A,(W'_B)_{B<A}\big)
    =Z\big(U_A\mid U_A+W_A,(W_B)_{B<A}\big)\;,
  \end{align*}
  as claimed.
\end{proof}

As for the proof of Lemma~\ref{cor:poly}, we need some preparation first.
Consider the $n$-dimensional vector space
\begin{align*}
  \mathcal{E} := \Span\{E_z\}_{z\in\{0,1\}^m}\;,
\end{align*}
whose elements are random variables that can be created as linear
combinations of $E_z$ over $\mathbb{F}_2$. Similarly, let
$\mathcal{E}_U:=\Span\{\mathcal{E}, U\}$ be the space of dimension $n+1$
spanned by $E_z$ and an additional independent uniform random variable $U$.
Since for a collection of random variables
$\mathcal{W}\subseteq\mathcal{E}_U$ their values are uniquely determined
by any subcollection that spans $\mathcal{W}$, the fact below is an
immediate consequence of Fact~\ref{fac:bhattacharyya-properties}.2:
\begin{fact}\label{fac:z-span}
  Let $\mathcal{W}\subseteq\mathcal{E}_U$.
  Then, $Z(U\mid \mathcal{W})=Z(U\mid\Span\{\mathcal{W}\})$.
\end{fact}

Since their dimensions are equal, of course the space $\mathcal{E}$
is isomorphic to the (multilinear)
polynomial space $\mathbb{F}_2[Z_1,\ldots,Z_m]$. Furthermore, it is easy
to check that $\{W_A:A\subseteq[m]\}$ is a basis of $\mathcal{E}$ which
lets us define a natural isomorphism
$\phi$ by
\begin{align}\label{eq:22}
  \phi W_A:=P_A\;.
\end{align}
Furthermore, let $\tau$ be a permutation of $\{0,1\}^m$. Permutation $\tau$
naturally induces a linear operator on $\mathcal{E}$, which we can give as
\begin{align}\label{eq:24}
  \tau W:=\phi^{-1}\left((\phi W)\circ\tau\right)\;,
\end{align}
in other words the linear operator $\tau$ satisfies the relation
$\phi\tau W=(\phi W)\circ\tau$.
The operator $\tau$ can be extended to $\mathcal{E}_U$ by adding
\begin{align*}
  \tau U:=U\;.
\end{align*}
Note that we have:
\begin{fact}\label{fac:tau-basis}
  $\tau E_z=E_{\tau^{-1}(z)}$.
\end{fact}
\begin{proof}
  First, we use~\eqref{eq:22} to verify that
  \begin{align*}
    (\phi E_z)(z')=1\iff z'=z\;.
  \end{align*}
  But from this it follows that
  \begin{align*}
    (\phi\tau E_z)(z')=((\phi E_z)\circ\tau)(z')
    =(\phi E_z)(\tau(z'))=1
    \iff \tau(z')=z\iff z'=\tau^{-1}(z)\;,
  \end{align*}
  and consequently $\tau E_z=E_{\tau^{-1}(z)}$.
\end{proof}
Using Fact~\ref{fac:tau-basis} we can prove a reformulation of
Fact~\ref{fac:permutation}:
\begin{fact}\label{fac:z-permutation}
  Let $\mathcal{W}\subseteq\mathcal{E}_U$ and let
  $\tau\mathcal{W}:=\{\tau W: W\in\mathcal{W}\}$. Then,
  \begin{align*}
    Z(U\mid\mathcal{W})=Z(U\mid\tau\mathcal{W})\;.
  \end{align*}
\end{fact}
\begin{proof}
  In our notation the conclusion looks identical to~\eqref{eq:23}
  in Fact~\ref{fac:permutation}, but we need to make sure the meanings
  of $\mathcal{W}$ and $\tau\mathcal{W}$ are the same in both settings.

  Indeed, take $W\in\mathcal{W}\subseteq\mathcal{E}_U$. Then, it must be
  \begin{align*}
    W = U+Ev^T\lor W=Ev^T
  \end{align*}
  for some $v\in\mathbb{F}_2^{\{0,1\}^m}$. On the other hand, taking matrix
  $P$ from~\eqref{eq:18} for permutation $\tau^{-1}$ and noting that
  \begin{align*}
    EP = (E_{\tau^{-1}(z)})_{z\in\{0,1\}^m}\;,
  \end{align*}
  from Fact~\ref{fac:tau-basis} it follows that if $W=U+Ev^T$, then
  \begin{align*}
    \tau W=U+\sum_z v_zE_{\tau^{-1}(z)}=U+EPv^T\;, 
  \end{align*}
  and similarly for $W=Ev^T$. Therefore,
  $Z(U\mid\mathcal{W})=Z(U\mid\tau\mathcal{W})$ holds by an application
  of Fact~\ref{fac:permutation} for permutation $\tau^{-1}$.
\end{proof}

Finally, we are ready to prove Lemma~\ref{cor:poly}:
\begin{proof}[Proof of Lemma~\ref{cor:poly}]
  Using~\eqref{eq:22} and~\eqref{eq:24}, we can rewrite the assumptions:
  \begin{align*}
    P_A\circ\tau\in P_B+\Span\{P_{<B}\}
    &\implies \tau W_A\in W_B+\Span\{W_{<B}\}\;,\\
    P_{<A}\circ\tau\subseteq\Span\{P_{<B}\}
    &\implies\tau W_{<A}\subseteq\Span\{W_{<B}\}\;.
  \end{align*}
  We also note that $\tau W_A\in W_B+\Span\{W_{<B}\}$ implies
  \begin{align*}
    U+\tau W_A\in (U+W_B)+\Span\{W_{<B}\}\;.
  \end{align*}
  But now, applying Corollary~\ref{cor:z-formula} (twice),
  Fact~\ref{fac:z-span} (also twice), Fact~\ref{fac:z-permutation},
  and Fact~\ref{fac:bhattacharyya-properties}.1,
  \begin{align*}
    Z_A
    &=Z(U\mid U+W_A,W_{<A})\\
    &=Z(U\mid\Span\{U+W_A,W_{<A}\})\\
    &=Z(U\mid\Span\{U+\tau W_A, \tau W_{<A}\})\\
    &\ge Z(U\mid\Span\{U+W_B, W_{<B}\})\\
    &=Z(U\mid U+W_B,W_{<B}) = Z_B\;.\qedhere
  \end{align*}
\end{proof}

\subsection{Proof of Lemma~\ref{lem:domination-equivalence}}\label{sec:za-zb}

\begin{enumerate}
\item Assume first that $a_i\le b_i$ for every $i$. Let
  $A':=A\setminus(A\cap B), B':=B\setminus(A\cap B)$
  and $A'=\{a'_1,\ldots,a'_k\},B'=\{b'_1,\ldots,b'_k\}$
  with $a'_1<\ldots<a'_k$ and $b'_1<\ldots<b'_k$. It is not difficult
  to see that we have $a'_i<b'_i$ for every $i$.
  Therefore, $B$ can be constructed from $A$ by $k$ applications
  of Rule 1.
  
  On the other hand, assume that $B$ can be constructed from $A$. Clearly,
  $B$ must have been obtained from $A$ only by applications
  of Rule~1. Let us say that there have been $k$ applications.
  If in the $i$-th application an element $c$ was removed and 
  $d$ was inserted, we will say that $c$ was \emph{mapped}
  to $d$.
  
  By an easy inductive argument, whenever $B$ can be constructed from $A$
  using a sequence of mappings $(c_1,d_1,\ldots,c_k,d_k)$, it can also
  be constructed using another sequence $(c'_1,d'_1,\ldots,\allowbreak c'_{k'},d'_{k'})$
  such that no element is used more than once, i.e.,
  $|\{c'_1,d'_1,\ldots,c'_{k'},d'_{k'}\}|=2k'$. Let us assume this property
  from now on.
  
  This means that we can write $A'=A\setminus (A\cap B)$, 
  and $B'\setminus (A\cap B)$ such that
  $A'=\{c_1,\ldots,c_k\}$, $B'=\{d_1,\ldots,d_k\}$ and $c_i<d_i$.
  But this implies that also after sorting
  $A'=\{c'_1<\ldots<c'_k\}$, $B'=\{d'_1<\ldots<d'_k\}$ we have
  $c'_i<d'_i$, which in turn implies $a_i\le b_i$ for every $i$.
\item Clearly, $\widetilde{A}^{(k)}$ is constructed from $A$ by $k$
  applications of Rule~2, so if 
  $B$ can be constructed from $\widetilde{A}^{(k)}$, 
  then $B$ can be constructed from $A$.
  
  In the other direction, first note that whenever $B$ can be constructed from
  $A$, there always exists a sequence of rule applications transforming $A$
  to $B$ such that all applications of Rule~2 occur before
  any applications of Rule~1.
  We will now argue that if $B$ can be constructed from $A$
  and the first applied rule is Rule~2, then $B$ can be constructed
  from $\widetilde{A}$. That $B$ can be constructed from
  $\widetilde{A}^{(k)}$ follows by induction.
  
  To this end, assume that in this first Rule~2 application
  set $A':=A\setminus\{c_1,c_2\}\cup\{c'\}$ was obtained from $A$ with
  $c_1<c_2$. Note that by definition, $c_1\le a$, $c_2\le a'$ and $b\le c'$.
  Therefore $A'$ (and therefore also $B$)
  can be constructed from $\widetilde{A}$
  by applying Rule~1 three times: Mapping $c_1$ to $a$, $c_2$ to $a'$
  (or possibly $c_1$ to $a'$ if $c_2=a$)
  and $b$ to $c'$. 
  Special cases when $b=a$ or $c'=c_i$ are handled in a similar way.\qed
\end{enumerate}

\subsection{Proof of Lemma~\ref{cor:random-set-dominated}}
\label{sec:za-zb-details}

Before we proceed to the main proof, let us
make an approximation of~\eqref{eq:27} for $r=m/2+\Theta(\sqrt{m})$:
\begin{corollary}\label{cor:maximum-tail}
  Let $A$ be a uniform random set of size $r=m/2+\alpha\sqrt{m}$ and
  $d=\gamma\sqrt{m}$ such that
  $|\alpha|,\gamma\le m^{1/12}$. Then, 
  \begin{align}\label{eq:28}
    \Pr[\text{$A$ not $d$-good}]\le\exp\left(
    8\alpha\gamma-8\gamma^2+\frac{C}{m^{1/4}}
    \right)
  \end{align}
  for some universal constant $C>0$.
\end{corollary}

\begin{proof}
  First, we can assume wlog that $\gamma>\alpha$: If $\alpha<0$, then it holds since
  $\gamma\ge 0$, and if $\alpha\ge 0$, then,
  since $8\alpha\gamma-8\gamma^2=-8\gamma(\gamma-\alpha)\ge 0$
  for $0\le\gamma\le\alpha$, the right-hand side of~\eqref{eq:28} exceeds 1. Additionally, we can assume wlog that $m$ is large
  enough (say, $m\ge 100$), since for $m<100$
  the constant $C$ in~\eqref{eq:28} can be chosen large enough
  for the inequality to hold.

  Since $\gamma>\alpha$, we check that $2d+1>2\gamma\sqrt{m}>2\alpha\sqrt{m}=2r-m$
  and apply Corollary~\ref{cor:good-vs-random-walk}:
  \begin{align}
    \Pr[\text{$A$ not $d$-good}]
    &=\frac{\binom{m}{r-2d-1}}{\binom{m}{r}}
      =\frac{r!(m-r)!}{(r-2d-1)!(m-r+2d+1)!}
      =\prod_{i=0}^{2d}\frac{r-i}{m-r+i+1}
      \nonumber\\
    &=\prod_{i=0}^{2\gamma\sqrt{m}}\frac{m/2+\alpha\sqrt{m}-i}
      {m/2-\alpha\sqrt{m}+i+1}
      =\prod_{i=0}^{2\gamma\sqrt{m}}
      \frac{1+2\alpha/\sqrt{m}-2i/m}
      {1-2\alpha/\sqrt{m}+2(i+1)/m}\;.
      \label{eq:48}
  \end{align}
  We proceed to bound each term on the right-hand side of~\eqref{eq:48}.
  In the numerator we use $1+2\alpha/\sqrt{m}-2i/m\le\exp\left(2\alpha/\sqrt{m}-2i/m
  \right)$. In the denominator we first invoke $m\ge 100$ to observe that
  $1-2\alpha/\sqrt{m}+2(i+1)/m>1-2/m^{5/12}>1/2$. Then, we apply the inequality
  $1+x\ge\exp(x-x^2)$, which holds for every $x\ge -1/2$ to bound each of the terms
  $1-2\alpha/\sqrt{m}+2(i+1)/m$.
  
  Plugging these inequalities into~\eqref{eq:48} results in
  \begin{align*}
    \Pr[\text{$A$ not $d$-good}]
    &\le\exp\left(\sum_{i=0}^{2\gamma\sqrt{m}}
      \frac{4\alpha}{\sqrt{m}}-\frac{4i+2}{m}
      +\left(\frac{2\alpha}{\sqrt{m}}-\frac{2i}{m}-\frac{2}{m}\right)^2\right)\\
    &\le\exp\left(\sum_{i=0}^{2\gamma\sqrt{m}}
      \frac{4\alpha}{\sqrt{m}}-\frac{4i+2}{m}
      +O\left(\frac{\alpha^2}{m}+\frac{\gamma^2}{m}+\frac{1}{m^2}\right)\right)\\
    &\le\exp\left(\sum_{i=0}^{2\gamma\sqrt{m}}
      \frac{4\alpha}{\sqrt{m}}-\frac{4i+2}{m}
      +O\left(\frac{1}{m^{5/6}}\right)\right)
      \le\exp\left(
      8\alpha\gamma-8\gamma^2
      +O\left(\frac{1}
      {m^{1/4}}\right)
      \right)\;.\qedhere
  \end{align*}
\end{proof}

\begin{proof}[Proof of Lemma~\ref{cor:random-set-dominated}]
  Let $d_i:=\lfloor\gamma_i\sqrt{m}\rfloor$ for
  \begin{align*}
    \gamma_1:=\frac{\beta-\alpha}{3}\;,\qquad\qquad
    \gamma_2:=\frac{2\beta+\alpha}{3}\;.
  \end{align*}
  Since $d_1+d_2\le(\gamma_1+\gamma_2)\sqrt{m}=k$,
  by Lemma~\ref{lem:random-set-dominated} and union bound we have
  \begin{align}\label{eq:29}
    \Pr[\lnot(A\ll B)]\le
    \Pr[\text{$\overline{A}$ not $d_1$-good}]+    
    \Pr[\text{$B$ not $d_2$-good}]\;.
  \end{align}
  Recall that $|\overline{A}|=m/2-(\alpha+\beta)\sqrt{m}$
  and $|B|=m/2+\alpha\sqrt{m}$. After checking that the assumptions of
  Corollary~\ref{cor:maximum-tail} hold, we use it to 
  separately estimate the two terms in~\eqref{eq:29}:
  \begin{align*}
    \Pr\Big[\text{$\overline{A}$ not $d_1$-good}\Big]
    &\le\exp\left(-8(\alpha+\beta)\gamma_1-8\gamma_1^2+\frac{C}{m^{1/4}}\right)
      =\exp\left(-\frac{16}{9}(\beta-\alpha)(2\beta+\alpha)+\frac{C}{m^{1/4}}\right)\;,\\
    \Pr\Big[\text{$B$ not $d_2$-good}\Big]
    &\le\exp\left(8\alpha\gamma_2-8\gamma_2^2+\frac{C}{m^{1/4}}\right)
      =\exp\left(-\frac{16}{9}(\beta-\alpha)(2\beta+\alpha)+\frac{C}{m^{1/4}}\right)\;,\\
  \end{align*}
  and the result follows.
\end{proof}

\subsection{Proof of Theorem~\ref{thm:lower-bound}}
\label{sec:lower-bound-details}

In order to prove Theorem~\ref{thm:lower-bound}, we start with giving
its reformulation which does not mention the $H(\cdot)$ function and instead
speaks only about collections of sets. Given $0<R<1$, throughout we fix
$r=r(R,m)$ to be the smallest $r$ such that $\binom{m}{\le r}\ge Rn$.
\begin{theorem}\label{thm:lower-sets}
  For every $0<R,\eps<1$, there exist $\delta>0$ and $m_0$ such that
  for every $m>m_0$ there exists a collection of sets
  $\mathcal{B}\subseteq\binom{m}{\le r}$ such that:
  \begin{enumerate}
  \item $|\mathcal{B}|\ge\delta\binom{m}{\le r}$.
  \item Letting $\mathcal{A}:=\left\{A\subseteq[m]:\exists B\in\mathcal{B}
      \text{ s.t.~}A\ll B\right\}$, we have $|\mathcal{A}|\le \eps n$.
  \end{enumerate}
\end{theorem}

We first prove that Theorem~\ref{thm:lower-sets} implies
Theorem~\ref{thm:lower-bound} and then prove Theorem~\ref{thm:lower-sets}.

\begin{proof}[Theorem~\ref{thm:lower-sets} implies Theorem~\ref{thm:lower-bound}]
  Let $0<R,\eps<1$ and take $\delta$, $m_0$, $\mathcal{B}$ and $\mathcal{A}$
  given by Theorem~\ref{thm:lower-sets}.

  We define function $H$ by
  induction, adding sets to a collection $\mathcal{A}'$ and maintaining
  an invariant $A\in\mathcal{A}'$ implies $H(A)\ne 0$. We start with
  setting $\mathcal{A}':=\mathcal{A}$ and assigning $H(A):=1$ for all $A\in\mathcal{A}$. Then, as long as
  $\sum_{A\in\mathcal{A}'}H(A)<\eps n$, we:
  \begin{itemize}
  \item Pick a set $A$ which is a minimal element of the partial order
    $\ll$ restricted to $\mathcal{P}(m)\setminus\mathcal{A}$.
  \item Assign $H(A):=\min\left(1,\eps n-\sum_{B\in\mathcal{A}'}H(B)\right)$.
  \item Add $A$ to $\mathcal{A}'$.
  \end{itemize}
  Since $\eps n < n$, the algorithm described above terminates. After that
  happens, we assign $H(A):=0$ to all remaining sets.

  Let us verify the four conditions from Theorem~\ref{thm:lower-bound} in turn.
  First, we need to show that $A\ll B$ implies $H(A)\ge H(B)$.
  We show it by induction, proving that this implication holds
  at all stages of building $\mathcal{A}'$ from $\mathcal{A}$.
  At the beginning we have $\mathcal{A}'=\mathcal{A}$,
  satisfying $H(A)=1$ for $A\in\mathcal{A}$ and $H(A)=0$
  otherwise. If $A\ll B$ and $H(B)=0$, then $H(A)\ge H(B)$ obviously
  holds. On the other hand if $A\ll B$ and $H(B)=1$, then $B\in\mathcal{A}$
  and from definition of $\mathcal{A}$ also $A\in\mathcal{A}$ and
  $H(A)=1$, $H(B)\ge H(A)$.
  
  Consider now a step where set $A$ is added to $\mathcal{A}'$
  and the value $H(A)$ increases. By induction, we only need to check 
  inequalities involving $A$. If $A\ll B$, then $H(A)\ge H(B)$ holds
  by induction, since $H(A)$ only increased. If $B\ll A$, then
  by the fact that $A$ is minimal when added to $\mathcal{A}'$
  we have $B\in\mathcal{A}'$ and $H(B)\ge H(A)$.

  The second condition that there is at most one set
  with $H(A)\notin\{0,1\}$ follows by construction, and similarly
  with the third condition
  $\sum_{A}H(A)= \eps n$.
  Finally, the fourth condition
  $|\{A:H(A)=1\land|A|\le r\}|\ge\delta\binom{m}{\le r}$ follows since
  $|\mathcal{B}|\ge\delta\binom{m}{\le r}$ and $\mathcal{B}\subseteq\mathcal{A}$ with $H(A)=1$ for all
  $A\in\mathcal{A}$.
\end{proof}

\begin{proof}[Proof of Theorem~\ref{thm:lower-sets}]
  Let $0<R,\eps<1$. Throughout the proof we assume that $m>m_0(R,\eps)$
  so that everything is well-defined. By Fact~\ref{fac:binom-clt}, we
  have
  \begin{align*}
    r=\frac{m}{2}+\alpha\sqrt{m}+o(\sqrt{m})
    \qquad\text{for}\qquad
    \alpha:=\frac{\Phi^{-1}(R)}{2}\;.
  \end{align*}
  Let $k:=\lceil4\sqrt{m}\rceil$ and recall the notation from
  Lemma~\ref{lem:run-plus-one}: $S_i(A)=2|A\cap[i]|-i$, 
  $\ell=\lfloor m/2\rfloor$ and
  \begin{align*}
    \mathcal{B}_{m,r,s}:=\left\{B\in\binom{m}{r}:S_{\ell}(B)>s\right\}\;.
  \end{align*}
  Furthermore, choose some $\gamma',\gamma>0$, such that
  \begin{align}\label{eq:33}
    \Phi(-2\alpha+8-2\gamma')\le\frac{\eps}{3}\;,
    \qquad\qquad
    \Phi\big(\sqrt{2}(-\gamma+2\gamma')\big)\le\frac{\eps}{3}\;,
  \end{align}
  and let
  \begin{align*}
    s:=\lceil\gamma\sqrt{m}\rceil\;,\qquad\qquad
    \mathcal{B}_0:=\mathcal{B}_{m,r-k,s}\qquad\qquad
    \mathcal{B}:=\left\{A\in\binom{m}{\le r}:\exists B\in\mathcal{B}_0
    \text{ s.t.~}A\ll B\right\}\;.
  \end{align*}
  We are going to argue that $\mathcal{B}\subseteq\binom{m}{\le r}$ is the
  collection of sets satisfying the two conditions of the theorem. To that
  end, we divide the rest of the proof into two claims.

  \begin{claim}
    There exists $\delta=\delta(R,\eps)>0$ such that
    $|\mathcal{B}|\ge\delta\binom{m}{\le r}$.
  \end{claim}
  \begin{proof}
    Observe that if a set $B\subseteq[m]$ satisfies
    the conditions
    \begin{align}
      \frac{\ell}{2}+\frac{s}{2}
      &<|B\cap[\ell]|<\frac{\ell}{2}+\frac{s}{2}+\sqrt{m}\;,
        \label{eq:31}\\
      \frac{m-\ell}{2}-\frac{s}{2}+\alpha\sqrt{m}-3\sqrt{m}
      &<|B\cap\{\ell+1,\ldots,m\}|
        <\frac{m-\ell}{2}-\frac{s}{2}+\alpha\sqrt{m}-2\sqrt{m}\;,
        \label{eq:32}
    \end{align}
    then, on the one hand, we have
    \begin{align*}
      S_{\ell}(B)=2|B\cap[\ell]|-\ell>s\;,
    \end{align*}
    and on the other hand
    \begin{align*}
      r-k<\frac{m}{2}+\alpha\sqrt{m}-3\sqrt{m}<|B|
      <\frac{m}{2}+\alpha\sqrt{m}-\sqrt{m}<r\;,
    \end{align*}
    therefore $B\in\mathcal{B}_{m,r-k',s}$ for some $0\le k'\le k$.
    Since, by multiple applications of
    the left containment in
    Lemma~\ref{lem:run-plus-one}, we have
    \begin{align*}
      \bigcup_{k'=0}^k \mathcal{B}_{m,r-k',s}\subseteq\mathcal{B}\;,
    \end{align*}
    it also holds that $B\in\mathcal{B}$. Therefore, letting $\mathcal{C}$
    to be the collection of all sets satisfying~\eqref{eq:31} and~\eqref{eq:32},
    we can use Fact~\ref{fac:binom-clt} to estimate
    \begin{align*}
      |\mathcal{B}|
      &\ge|\mathcal{C}|
        =\left[\binom{\ell}{<\left(\frac{\ell}{2}+\frac{s}{2}+\sqrt{m}\right)}
        -\binom{\ell}{\le\left(\frac{\ell}{2}+\frac{s}{2}\right)}\right]\\
      &\qquad\qquad\qquad\cdot
        \left[
        \binom{m-\ell}{<\left(
        \frac{m-\ell}{2}-\frac{s}{2}+\alpha\sqrt{m}-2\sqrt{m}
        \right)}
        -\binom{m-\ell}{\le\left(
        \frac{m-\ell}{2}-\frac{s}{2}+\alpha\sqrt{m}-3\sqrt{m}
        \right)}\right]\\
      &=\Big[\Phi\big(\sqrt{2}(\gamma+2)\big)-\Phi\big(\sqrt{2}\gamma\big)\Big]
        \cdot\Big[
        \Phi\big(\sqrt{2}(-\gamma+2\alpha-4)\big)
      -\Phi\big(\sqrt{2}(-\gamma+2\alpha-6)\big)
        \Big]\cdot n+o(n)\\
      &>\delta n>\delta\binom{m}{\le r}\;.\qedhere
    \end{align*}
  \end{proof}
  
  \begin{claim}
    Let
    $\mathcal{A}:=\left\{A\subseteq[m]:\exists B\in\mathcal{B}\text{ s.t.~}
      A\ll B\right\}$. Then, $|\mathcal{A}|\le\eps n$.
  \end{claim}
  \begin{proof}
    By definition of $\mathcal{B}$, we have
    \begin{align*}
      \mathcal{A}=\left\{A\subseteq[m]:\exists B\in\mathcal{B}_0\text{ s.t.~}
      A\ll B\right\}\;.
    \end{align*}
    Observe that 
    $\mathcal{A}\cap\binom{m}{r-k}=\mathcal{B}_0=\mathcal{B}_{m,r-k,s}$. Hence, by
    the right containment in Lemma~\ref{lem:run-plus-one}, for every $i\ge 0$,
    \begin{align*}
      \mathcal{A}\cap\binom{m}{r-k+i}\subseteq\mathcal{B}_{m,r-k+i,s-2i}
    \end{align*}
    and consequently
    \begin{align*}
      |\mathcal{A}|
      &\le
        \sum_{i=0}^{m-r+k}|\mathcal{B}_{m,r-k+i,s-2i}|\;.
    \end{align*}
    Letting $i_0:=\lceil\gamma'\sqrt{m}\rceil$, the last sum can be estimated
    as
    \begin{align*}
      \sum_{i=0}^{m-r+k}|\mathcal{B}_{m,r-k+i,s-2i}|
      &\le\big|\big\{
        A: S_{\ell}(A)>s-2i_0
        \big\}\big|
      +\big|\big\{
        A: |A|>r-k+i_0
        \big\}\big|\;.
    \end{align*}
    To analyze the two terms above, we use
    Fact~\ref{fac:binom-clt} for the last time:
    \begin{align*}
      \big|\big\{A: S_{\ell}(A)>s-2i_0\big\}\big|
      &=
        \binom{\ell}{>(\frac{\ell}{2}+\frac{s}{2}-i_0)}2^{m-\ell}
        =\Phi\big(\sqrt{2}(-\gamma+2\gamma')\big)\cdot n + o(n)\;,\\
      \big|\big\{A:|A|>r-k+i_0\big\}\big|
      &=\binom{m}{>(r-k+i_0)}
        =\Phi\big(-2\alpha+8-2\gamma'\big)\cdot n+o(n)\;.
    \end{align*}
    Recalling~\eqref{eq:33}, we conclude
    \begin{align*}
      |\mathcal{A}|
      &\le\Big(
      \Phi\big(\sqrt{2}(-\gamma+2\gamma')\big)
      +\Phi\big(-2\alpha+8-2\gamma'\big)
      \Big)\cdot n+o(n)
      <\eps n\;.\qedhere
    \end{align*}
  \end{proof}
\phantom{\qedhere}\end{proof}

\paragraph{Acknowledgment} We are grateful to the referees for
their useful corrections and suggestions.

\bibliographystyle{alpha}
\bibliography{bibliography_arxiv_v2}

\end{document}